\documentclass[11pt,letterpaper]{article}

\pdfoutput=1

\usepackage[letterpaper,margin=1in]{geometry}
\usepackage[T1]{fontenc}
\usepackage[utf8]{inputenc}
\usepackage{lmodern}
\usepackage{microtype}

\usepackage{amsmath,amssymb,amsthm,mathtools}
\usepackage{enumitem}
\usepackage{xcolor}
\usepackage{comment}
\usepackage{todonotes}
\usepackage[hidelinks]{hyperref}
\usepackage[capitalize,nameinlink,noabbrev]{cleveref}
\usepackage[mathlines]{lineno}
\usepackage{booktabs}
\usepackage{longtable}
\usepackage{array}
\usepackage{pdflscape}
\usepackage{caption}

\newcolumntype{L}[1]{>{\raggedright\arraybackslash}p{#1}}

\usepackage{tikz}
\usetikzlibrary{arrows.meta,positioning,calc,backgrounds}

\allowdisplaybreaks

\newtheorem{theorem}{Theorem}[section]
\newtheorem{lemma}[theorem]{Lemma}
\newtheorem{proposition}[theorem]{Proposition}

\theoremstyle{definition}
\newtheorem{definition}[theorem]{Definition}
\theoremstyle{remark}
\newtheorem{remark}[theorem]{Remark}

\newcommand{\HLLDEC}{\mathrm{HLLDEC}}
\newcommand{\HLLSearch}{\mathrm{HLLSearch}}
\newcommand{\HDDLDEC}{\mathrm{HDDLDEC}}
\newcommand{\HDDLSearch}{\mathrm{HDDLSearch}}
\newcommand{\ASEGDEC}{\mathrm{ASEGDEC}}
\newcommand{\UOR}{\mathrm{UOR}}
\newcommand{\Dom}{\operatorname{Dom}}
\newcommand{\bits}{\{0,1\}}

\newcommand{\pp}{\perp}
\newcommand{\Advpm}{\mathrm{Adv}^{\pm}}
\newcommand{\ket}[1]{\lvert #1 \rangle}
\newcommand{\Good}{\mathsf{Good}}

\newcommand{\Live}{\mathsf{Live}}

\providecommand{\HDDLDEC}{\mathrm{HDDLDEC}}
\providecommand{\HDDLSearch}{\mathrm{HDDLSearch}}

\providecommand{\UOR}{\mathrm{UOR}}

\definecolor{Ocean}{HTML}{2A6F97}
\definecolor{SoftBlue}{HTML}{EAF4FB}
\definecolor{NodeBlue}{HTML}{D6EAF8}
\definecolor{DeadEdge}{HTML}{6E8FA8}
\definecolor{MarkRed}{HTML}{C93C4B}
\definecolor{SoftRed}{HTML}{FDEBED}

\title{\textbf{Quantum Query Complexity for List Search}}
\author{
  Niranka Banerjee\thanks{Department of Information Engineering, Mie University, Japan. Email: \texttt{banerjee@eng.mie-u.ac.jp}.}
  \and
  Akinori Kawachi\thanks{Department of Information Engineering, Mie University, Japan. Email: \texttt{kawachi@info.mie-u.ac.jp}.}
}
\date{}

\setlist[itemize]{leftmargin=1.5em,itemsep=0.25em,topsep=0.25em}
\setlist[enumerate]{leftmargin=1.8em,itemsep=0.25em,topsep=0.25em}

\begin{document}

\maketitle
\thispagestyle{empty}
\begin{abstract}
Searching in a linked list is one of the most basic problems in classical algorithms. Although the
nodes of the list come with memory addresses, classically those addresses play no role in the cost of
search: one simply starts at the head and follows successor pointers to search for an element. In this paper, we show that
the quantum setting is different. Here, the ambient address space from which the list vertices are
drawn can itself affect the query complexity.

We study the following problem analogous to search in a linked list in the query complexity model: the input consists of an address universe $[N]$, a public start
symbol $s$, a successor oracle $f$ whose non-$\perp$ values trace a hidden simple path
$s \to a_1 \to a_2 \to \cdots \to a_\ell \to \perp$,
and a marking oracle $g$ that marks at most one list vertex. The task is to decide whether the
list contains a marked vertex.

We prove that both the decision and search versions of this problem for all $N \ge \ell \ge 1$ have quantum query complexity
\[
\Theta\!\bigl(\min\{\ell,(N\ell)^{1/4}\}\bigr).
\]
Thus, quite surprisingly, when $N < \ell^3$  the optimal quantum complexity is $(N\ell)^{1/4}$, which is
strictly smaller than the $\Theta(\ell)$ cost of ordinary linked-list traversal.  This gives a precise characterization of when the ambient
address space yields a genuine quantum advantage for linked-list search.
  
 We extend our results and give the same tight asymptotic bounds for the natural double linked-list version as well. 
\end{abstract}




\section{Introduction}
  Linked lists are among the most basic pointer-based data structures in classical computing.    If one wants to search a linked list, depending on whether we have a single or double linked list the natural strategy is simply to start at the head and
follow the next pointers or next and previous pointers until the desired item is found or the list ends. In a pointer-machine view, a linked list is not given as an explicit ordered sequence of vertices. It is given by memory addresses and pointer fields. Thus the list vertices are hidden inside an ambient address space.

From this classical
point of view though, the actual numerical addresses of the list nodes carry no asymptotic information:
whether the list occupies the first $\ell$ memory locations or is hidden inside a much larger
ambient universe $[N]=\{1,\dots,N\}  $, one still has to reveal the list sequentially, one pointer at a time. Even when the list
vertices are hidden inside a larger address space, the ambient universe parameter $N$ is classically
useless, and the complexity remains $\Theta(\ell)$; see
Proposition~\ref{prop:classical-baseline}.

A quantum algorithm can potentially exploit the
ambient universe itself, not just the local successor and predecessor structure of the hidden list. This leads
to the basic question addressed in this paper:

\begin{quote}
Can an $\ell$-vertex linked list inside a large ambient universe $[N]$ use that ambient space to make search
genuinely easier quantumly, even though the same ambient space is irrelevant classically?
\end{quote}

Our main result characterizes exactly how the two parameters $N$ and $\ell$ interact in the quantum oracle query model, which is the standard formalization for search problems in quantum computing. We now define the corresponding list-search problem on the query model.




\begin{definition}[Hidden linked-list decision and search] \label{def:def1}
Fix integers $N \ge \ell \ge 1$. An input to $\HLLDEC(N,\ell)$ consists of oracle access to a
successor function
\[
f : \{s\}\cup[N] \to [N]\cup\{\pp\}
\]
and a marking function
\[
g : [N]\to\bits,
\]
with the promise that there exist pairwise distinct vertices $a_1,\dots,a_\ell\in[N]$, which we call \emph{list vertices}, such that
\[
f(s)=a_1,\qquad
f(a_i)=a_{i+1}\ \ (1\le i<\ell),\qquad
f(a_\ell)=\pp,
\]
and \(f(x)=\bot\) for every address
\(x\in [N]\setminus\{a_1,\ldots,a_\ell\}\).

The marking promise is
\[
\sum_{i=1}^{\ell} g(a_i)\le 1,
\]
and $g(x)=0$ off the list. The output of $\HLLDEC(N,\ell)$ is $1$ if and only if some $a_i$ is marked.

 The same model can also be realized when addresses outside the list have arbitrary successors and not $\bot$. Suppose that each address $x\in[N]$ contains, in addition to its ordinary data and pointer fields, a tag $\tau(x)$, initially equal to $0$. When a list is created, it is assigned a nonzero identifier $L$, known to the algorithm, and every vertex $a_i$ belonging to the list is assigned the tag $\tau(a_i)=L$. The contents of every address with $\tau(x)\neq L$, including its pointer field, may then be arbitrary.

Throughout the paper, oracle access to $(f,g)$ is quantum oracle access in the
standard reversible sense. Formally, we may view $(f,g)$ as a single
finite-alphabet input whose query coordinates are the disjoint union of the
successor coordinates $(F,x)$, where $x\in\{s\}\cup[N]$, and the marking
coordinates $(G,x)$, where $x\in[N]$. The corresponding oracle acts by
$O_{f,g}|F,x,z\rangle = |F,x,z\oplus f(x)\rangle$ and 
 $O_{f,g}|G,x,b\rangle = |G,x,b\oplus g(x)\rangle$.
Here $f(x)\in [N]\cup\{\bot\}$ is encoded using a fixed binary encoding, and
$g(x)\in\{0,1\}$. Thus one query means one application of this combined oracle;
in particular, the query register may be in superposition over successor and
marking coordinates. Asking for the successor of a single address costs one
query, and asking whether a single address is marked costs one query.



In particular, asking for the successor of one address costs one query, and asking
whether one address is marked costs one query.  All algorithms may use arbitrary
input-independent unitaries between such oracle calls.

The search version $\HLLSearch(N,\ell)$ asks for the unique marked vertex if one exists, and for a failure symbol otherwise. Note that the public start symbol $s$ is outside $[N]$; this is mainly for ease of exposition.
\end{definition}

\begin{definition}[Hidden double linked-list decision and search]
\label{def:dbl-hddl}

Fix integers $N\geq \ell\geq 1$. An input to
$\mathrm{HDDLDEC}(N,\ell)$ consists of oracle access to
\[
\operatorname{next}:\{s\}\cup[N]\to [N]\cup\{\bot\},
\qquad
\operatorname{prev}:\{s\}\cup[N]\to \{s\}\cup[N]\cup\{\bot\},
\qquad
g:[N]\to\{0,1\},
\]
with the promise that there are pairwise distinct vertices
$a_1,\ldots,a_\ell\in[N]$ such that
\[
\operatorname{next}(s)=a_1,\qquad
\operatorname{next}(a_i)=a_{i+1}\quad (1\leq i<\ell),
\qquad
\operatorname{next}(a_\ell)=\bot,
\]
and
\[
\operatorname{prev}(s)=\bot,\qquad
\operatorname{prev}(a_1)=s,\qquad
\operatorname{prev}(a_i)=a_{i-1}\quad (2\leq i\leq \ell).
\]
Every off-list address is dead in both directions:
\[
\operatorname{next}(x)=\operatorname{prev}(x)=\bot
\qquad
\text{for every }x\in[N]\setminus\{a_1,\ldots,a_\ell\}.
\]
The marking promise is
\[
\sum_{i=1}^{\ell} g(a_i)\leq 1,
\qquad
g(x)=0 \text{ off the list}.
\]
The output of $\mathrm{HDDLDEC}(N,\ell)$ is $1$ if and only if some
list vertex is marked.

Again, oracle access to $(\operatorname{next},\operatorname{prev},g)$ is quantum oracle access in the standard reversible sense. Formally, we may view $(\operatorname{next},\operatorname{prev},g)$ as a single finite-alphabet input whose query coordinates are the disjoint union of the successor coordinates $(\mathrm{Next},x)$, where $x\in \{s\}\cup[N]$, the predecessor coordinates $(\mathrm{Prev},x)$, where $x\in \{s\}\cup[N]$, and the marking coordinates $(G,x)$, where $x\in[N]$. The corresponding oracle acts by
$O_{\operatorname{next},\operatorname{prev},g}
\ket{\mathrm{Next},x,z}
=
\ket{\mathrm{Next},x,z\oplus \operatorname{next}(x)}$,
$O_{\operatorname{next},\operatorname{prev},g}
\ket{\mathrm{Prev},x,z}
=
\ket{\mathrm{Prev},x,z\oplus \operatorname{prev}(x)}$,
and
$O_{\operatorname{next},\operatorname{prev},g}
\ket{G,x,b}
=
\ket{G,x,b\oplus g(x)}$.
Here $\operatorname{next}(x)\in [N]\cup\{\bot\}$, $\operatorname{prev}(x)\in \{s\}\cup[N]\cup\{\bot\}$, and $g(x)\in\{0,1\}$ are encoded using fixed binary encodings. Thus one query means one application of this combined oracle; in particular, the query register may be in superposition over successor, predecessor, and marking coordinates. Asking for the successor of a single address costs one query, the predecessor of a single address costs one query and whether a single address is marked costs one query. All algorithms may use arbitrary input-independent unitaries between such oracle calls.


The search problem $\mathrm{HDDLSearch}(N,\ell)$
asks to output the unique marked list vertex if it exists, and a failure
symbol otherwise.

\end{definition}

Let $D(F)$ denote the deterministic query complexity of a decision or search problem $F$.  Let $R(F)$ and $Q(F)$ denote its bounded-error randomized and quantum query complexities, respectively; that is, the minimum number of queries required by a randomized or quantum algorithm that succeeds with probability at least $2/3$ on every valid input.  We first show that, classically, the ambient space is indeed inert.

\begin{proposition}\leavevmode\footnote{The proof is deferred to Appendix~\ref{app:propclass}.}
\label{prop:classical-baseline}
For all integers $N\ge \ell\ge 1$,
\[
D(\HLLDEC(N,\ell)), 
R(\HLLDEC(N,\ell)),
D(\HLLSearch(N,\ell)),
R(\HLLSearch(N,\ell))
\in
\Theta(\ell)\]
\[D(\HDDLDEC(N,\ell)), 
R(\HDDLDEC(N,\ell)),
D(\HDDLSearch(N,\ell)),
R(\HDDLSearch(N,\ell))
\in
\Theta(\ell)\]
\end{proposition}

Proposition~\ref{prop:classical-baseline} formalizes the classical intuition: even when the
list vertices are drawn from an ambient universe of size $N$, the parameter $N$ plays no
asymptotic role. Quantumly, by contrast, there are already two immediate upper bounds.

First, one can ignore the ambient space and simply traverse the hidden list from the public
start symbol $s$, which costs $O(\ell)$ queries. Second, one can ignore the pointer structure
and search the ambient universe $[N]$ directly with Grover's algorithm~\cite{Grover97}:
because $g(x)=0$ off the list and at most one list vertex is marked, Grover search over
$[N]$ finds the marked vertex, if it exists, in $O(\sqrt N)$ queries. Thus
\[
Q(\HLLDEC(N,\ell)),\; Q(\HLLSearch(N,\ell))
  \in O(\min\{\ell,\sqrt N\}).
\]

Similar bounds also follow naively for double linked lists. Already this naive bound shows a genuinely quantum role for the ambient space. When
$N<\ell^2$, it gives $o(\ell)$, whereas Proposition~\ref{prop:classical-baseline} shows that the
classical complexity is tightly $\Theta(\ell)$. But can we do better than this bound?


Our theorems give a positive answer, and it is tight.

\begin{theorem}\label{thm:main}
For every $N \ge \ell \ge 1$,
\[
Q(\HLLDEC(N,\ell)),\; Q(\HLLSearch(N,\ell))
  \in \Theta\!\left(\min\{\ell,(N\ell)^{1/4}\}\right).
\]
\end{theorem}

\begin{theorem}\label{thm:main1}
For every $N \ge \ell \ge 1$,
\[
Q(\HDDLDEC(N,\ell)),\; Q(\HDDLSearch(N,\ell))
  \in \Theta\!\left(\min\{\ell,(N\ell)^{1/4}\}\right).
\]
\end{theorem}

These theorems identify a sharp threshold at $N\asymp \ell^3$. When the ambient universe is comparatively
small, namely $N\lesssim \ell^3$, the optimal complexity is $\Theta((N\ell)^{1/4})$. When
$N\gtrsim \ell^3$, the problem is dominated by the inherently sequential cost of moving
along the list, and simple traversal in $\Theta(\ell)$ queries is optimal.

There is also a broader motivation for the model. Much of quantum search theory lies at one
of two extremes. At one extreme is unstructured search, typified by Grover search and
amplitude amplification, where the search space is explicit and there is essentially no internal
geometry to exploit~\cite{Grover97,BHMT02}. At the other extreme are models with explicit
structure, such as ordered search~\cite{HNS02,CLP07,ChildsLee08,BDW99,CCMS},
spatial search on an explicitly given graph~\cite{AaronsonAmbainis2005,ChildsGoldstone04,Ben02},
or quantum backtracking on an explicit rooted tree~\cite{Montanaro18}. Consider, for example \emph{spatial search}, where the search space is the vertex set of an explicit graph (a path may represent a linked list) and the algorithm is constrained by the graph geometry. In this setting, quantum walks can still yield substantial speedups: Aaronson and Ambainis introduced a model of quantum search on graphs and showed that, on sufficiently well-connected spatial structures, one can achieve near-Grover performance despite locality constraints \cite{AaronsonAmbainis2005}. Interestingly, they stated that if the explicit graph is a path then we cannot do better than traversing the entire length of the path.

The hidden
linked-list model sits strictly between these two regimes. The ambient universe $[N]$ is
unstructured, but inside it lies a structured searchable subspace, namely a path of length
$\ell$, and that structured subspace is itself hidden. The marked element, if it exists, lies
inside this hidden path. An algorithm must therefore combine two tasks: it must first discover
useful support inside an unstructured universe, and then exploit the local successor structure
within that support. In this sense, hidden linked-list search isolates a basic intermediate
regime between unstructured quantum search and fully explicit structured search.  Figure~\ref{fig:search-regimes} summarizes this viewpoint pictorially. 

This perspective also points to a direction of future research beyond linked lists. 
One can imagine other, (especially pointer based) data structures whose natural model is such a structured
support hidden inside a much larger ambient universe. Classically, any relevant operation may run in time polynomial in the size 
of the support itself. The quantum question is whether the ambient universe can 
be used more directly: {\emph can a quantum algorithm exploit coherent access to 
$[N]$ to reduce the cost of various operations on such structured 
supports?}

\begin{figure}[t]
\centering
\begin{tikzpicture}[font=\small, line width=0.9pt]


\draw[draw=blue!60!black, fill=blue!5] (0,0) circle (1.9);
\node[font=\bfseries\scriptsize, text=blue!70!black] at (0,1.62) {$[N]$};

\filldraw[draw=orange!80!black, fill=orange!30] (-1.10, 0.75) rectangle ++(0.22,0.22);
\filldraw[draw=orange!80!black, fill=orange!30] (-0.30, 1.00) rectangle ++(0.22,0.22);
\filldraw[draw=red!75!black,    fill=red!30]      ( 0.80, 0.65) rectangle ++(0.22,0.22);
\filldraw[draw=orange!80!black, fill=orange!30] ( 1.00,-0.10) rectangle ++(0.22,0.22);
\filldraw[draw=orange!80!black, fill=orange!30] ( 0.35,-0.95) rectangle ++(0.22,0.22);
\filldraw[draw=orange!80!black, fill=orange!30] (-0.65,-1.00) rectangle ++(0.22,0.22);
\filldraw[draw=orange!80!black, fill=orange!30] (-1.15,-0.20) rectangle ++(0.22,0.22);
\filldraw[draw=orange!80!black, fill=orange!30] ( 0.00, 0.05) rectangle ++(0.22,0.22);

\node[font=\bfseries\small] at (0,-2.45) {Unstructured search};
\node[font=\scriptsize, text=black!70] at (0,-2.85) {Grover};


\filldraw[draw=orange!80!black, fill=orange!30] (4.40,-0.11) rectangle ++(0.22,0.22);
\filldraw[draw=orange!80!black, fill=orange!30] (5.30,-0.11) rectangle ++(0.22,0.22);
\filldraw[draw=red!75!black,    fill=red!30]      (6.20,-0.11) rectangle ++(0.22,0.22);
\filldraw[draw=orange!80!black, fill=orange!30] (7.10,-0.11) rectangle ++(0.22,0.22);
\filldraw[draw=orange!80!black, fill=orange!30] (8.00,-0.11) rectangle ++(0.22,0.22);

\draw[->, draw=teal!70!black] (4.62,0) -- (5.30,0);
\draw[->, draw=teal!70!black] (5.52,0) -- (6.20,0);
\draw[->, draw=teal!70!black] (6.42,0) -- (7.10,0);
\draw[->, draw=teal!70!black] (7.32,0) -- (8.00,0);

\node[font=\bfseries\small] at (6.30,-2.45) {Fully structured search};
\node[font=\scriptsize, text=black!70] at (6.30,-2.85) {explicit linked list};


\draw[draw=blue!60!black, fill=blue!5] (12.60,0) circle (1.9);
\node[font=\bfseries\scriptsize, text=blue!70!black] at (12.60,1.62) {$[N]$};

\filldraw[draw=orange!80!black, fill=orange!18] (11.55, 1.00) rectangle ++(0.20,0.20);
\filldraw[draw=orange!80!black, fill=orange!18] (12.20, 0.92) rectangle ++(0.20,0.20);
\filldraw[draw=orange!80!black, fill=orange!18] (13.25, 0.92) rectangle ++(0.20,0.20);
\filldraw[draw=orange!80!black, fill=orange!18] (13.55,-0.85) rectangle ++(0.20,0.20);
\filldraw[draw=orange!80!black, fill=orange!18] (12.35,-1.05) rectangle ++(0.20,0.20);
\filldraw[draw=orange!80!black, fill=orange!18] (11.30,-0.80) rectangle ++(0.20,0.20);

\filldraw[draw=orange!80!black, fill=orange!30] (11.34, 0.44) rectangle ++(0.22,0.22);
\filldraw[draw=orange!80!black, fill=orange!30] (11.89,-0.46) rectangle ++(0.22,0.22);
\filldraw[draw=orange!80!black, fill=orange!30] (12.59, 0.07) rectangle ++(0.22,0.22);
\filldraw[draw=red!75!black,    fill=red!30]      (13.24,-0.39) rectangle ++(0.22,0.22);
\filldraw[draw=orange!80!black, fill=orange!30] (13.69, 0.24) rectangle ++(0.22,0.22);

\draw[->, draw=teal!70!black] (11.56,0.55) -- (11.89,-0.24);
\draw[->, draw=teal!70!black] (12.11,-0.35) -- (12.59,0.18);
\draw[->, draw=teal!70!black] (12.81,0.18) -- (13.24,-0.28);
\draw[->, draw=teal!70!black] (13.46,-0.28) -- (13.69,0.35);

\node[font=\bfseries\small] at (12.60,-2.45) {Our model};
\node[font=\scriptsize, text=black!70] at (12.60,-2.85) {hidden linked list inside $[N]$};

\end{tikzpicture}
\caption{
Three search regimes.
Left: unstructured search over the ambient universe $[N]$.
Middle: fully structured search, where the linked list is explicit.
Right: our model, where $[N]$ is the ambient universe, but only a hidden linked list inside $[N]$ is relevant; other vertices of $[N]$ need not belong to that support.
In each panel, the red box indicates the unique marked node, if it exists, corresponding to the marking function $g$.
}
\label{fig:search-regimes}
\end{figure}
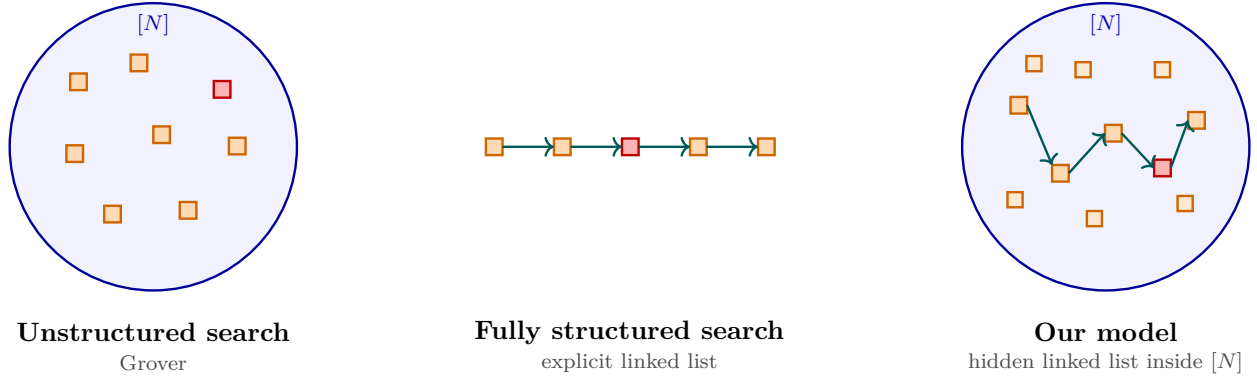

\paragraph*{Related Work:}Let us also contrast our model with the quantum ordered search model already studied in quantum literature. In the standard ordered-search problem, one is given oracle access to a monotone bit string
$x=(x_1,\dots,x_n)\in\{0,1\}^n$
with the promise that
$0=x_0 \le x_1 \le \cdots \le x_n = 1$,
and the goal is to output the unique index
$i^\star = \min\{i\in[n] : x_i = 1\}$.
Equivalently, this is the problem of locating a target in a sorted array, or searching in a
balanced binary search tree using comparison queries. Classically, the exact query complexity
is $\lceil \log_2 n\rceil$. Quantumly, this problem has been studied extensively, and the known
results show only constant-factor improvements over classical binary search: exact quantum
algorithms use $O(\log n)$ queries (with extensive research in reducing the constant factors)~\cite{HNS02,CLP07,ChildsLee08,BDW99,CCMS}. This remains true even if one embeds the ordered-search instance into a much larger
ambient address space, provided the hidden support still carries the usual comparison
structure of a balanced binary search tree. See Section~\ref{app:ordsearch} in the appendix for a short explanation.

Quantum-walk search gives another important point of comparison. In many quantum-walk algorithms, the algorithm starts with a graph, a Markov chain, or a local transition rule whose structure is available in advance, either as part of the problem formulation or as a walk chosen by the algorithm. This viewpoint appears, for example, in Szegedy's quantization of Markov-chain-based algorithms~\cite{Sze04}, the MNRS framework for search via quantum walks~\cite{MNRS11}, and Ambainis's quantum-walk algorithm for element distinctness~\cite{Amb07}.

Another related query problem of \emph{pointer chasing}, has also been studied in the quantum query setting recently, particularly in connection with parallel and adaptive queries~\cite{CarolanGilaniVempati25}. In their formulation, one is given oracle access to a function $X:\{0,1\}^n\to\{0,1\}^n$, viewed as assigning a pointer to every address, and the task is to determine a bit of $X^k(0^n)$ after following $k$ successive pointers. Rather than asking only for the result of a fixed number of pointer iterations, our problem however studies how global quantum access to the ambient address space interacts with search on a hidden linked structure.
 
\medskip

 We now briefly explain the main ideas behind both upper and lower bounds of the argument.

\paragraph{Proof overview.} 

 For the \emph{upper bound} of \emph{single linked lists}, the main idea is to combine global quantum access to the ambient universe \([N]\) with local exploration of the hidden list. The first step is to use the successor oracle to identify the nonterminal list vertices, namely those addresses whose successor is not $\bot$. Since there are exactly $\ell-1$ such vertices, exact amplitude amplification prepares the uniform superposition over them in \(O(\sqrt{N/\ell})\) queries \cite{BHMT02}. We then fix a parameter \(k\), inspect the first \(k\) vertices of the list directly, and, if the mark is not found there, search among the nonterminal vertices for one from which a forward walk  of length $k$ reaches the marked vertex. Checking this property costs $O(k)$ queries, and if the marked vertex lies beyond the initial prefix then at least $k$ starting vertices reach the marked vertex within $k$ successor steps.

A fixed-point search routine therefore finds such a starting point in $O(\sqrt{\ell/k})$ iterations, after which an additional $O(k)$-step walk recovers the marked vertex \cite{YoderLowChuang14}. Optimizing the resulting bound over $k$ gives $O((N\ell)^{1/4})$ when $N\le \ell^3$, while for $N>\ell^3$ simple traversal already gives the optimal $O(\ell)$ bound. For \emph{double linked lists} a similar algorithm works with appropriate changes to incorporate the predecessor oracle.

\medskip

The \emph{lower bounds} are the more delicate part of the paper.  At a high level, we would like to say that searching a hidden list is hard because the algorithm has to discover the relevant part of the address universe before it can search it. However, turning this intuition into a quantum lower bound requires some care. A quantum algorithm may query the ambient universe in superposition, and a lower-bound proof must show that no single query, even in superposition, can reveal too much information about where the marked vertex is hidden.

We prove the lower bounds using the negative-weight adversary method 
 \footnote{Readers may refer to
Section~\ref{prelims}, where we recall the negative weight adversary bound and explain how it applies to our oracle models.} \cite{HLS07}.  The proof has the same broad shape for single and double linked
lists, but the double linked-list case needs some modifications.  We first explain
the single linked-list lower bound.

\emph{Single linked lists:}
The first step is to isolate a small local problem.  Instead of working
immediately with the full hidden list $s \to a_1 \to \cdots \to a_\ell \to \bot$,
we look at one hidden segment.  Thus the segment has
the form $u_{\mathrm{in}} \to x_1 \to \cdots \to x_L \to u_{\mathrm{out}}$,
where \(u_{\mathrm{in}}\) and $u_{\mathrm{out}}$ are known public symbols, but
the internal vertices $x_1,\ldots,x_L$ are hidden addresses.  The task is to
decide whether the segment contains a marked vertex.  Later, this is formalized as the problem $\mathrm{ASEGDEC}(N,L)$.

The main technical statement for one such segment is the following:

\begin{quote}
A hidden single linked-list segment of length $L$ requires $\Omega(L)$ quantum queries when it lives in an
ambient address space of size $\Theta(L^3)$.
\end{quote}

Let us broadly describe why this is true.  We choose a family of layered paths as our hard instances.  The address space is divided into \(L\) layers, and the hidden segment
chooses one address from each layer.  The possible marked vertex is placed at
the end of the segment.  The adversary matrix of the negative weight adversary bound compares a yes-input, where the
last chosen vertex is marked, with a no-input, where the same kind of path has no marked vertex.

The weights in the adversary matrix are chosen according to where two hidden
paths first split.  If two paths split near the public entrance, then that pair
should not be given too much weight.  If the two paths agree for a long prefix and split only near the end those pairs are harder, and the adversary matrix gives them larger weight.

The important point is that this weighting works exactly at the cubic scale i.e. when the universe has size cubic in the segment length.
When the layer width is \(L^2\), the whole segment uses \(\Theta(L^3)\) possible
addresses.  At this scale, the total adversary weight is large, but every single
successor query or marking query sees only a controlled part of it.  This gives
the lower bound \(\Omega(L)\).

This also explains why the proof cannot simply take one segment of length
\(\ell\) for every value of \(N\).  If \(N<\ell^3\), then there is not enough space to give each layer width \(\ell^2\).  The same adversary construction does not work here.  

\medskip
The way around this is to keep each hard segment at its own cubic scale and
then put many such segments in series.  Choose a shorter segment length
\(L_0\) so that one segment of length \(L_0\) fits at its cubic scale.  In the
small-universe regime this means, up to constants,
$L_0 \approx \sqrt{N/\ell}$.
We then place many independent length-\(L_0\) boundary segments one after
another inside a single hidden linked list.

Now impose the promise that at most one segment contains a marked vertex.  Under
this promise, deciding whether the whole list contains a mark is just unique-OR
of the segment outputs: one must decide whether one of the many hidden segments
is positive.  The negative adversary bound composes under such block
composition \cite[Lemma~5.2]{LeeMittalReichardtSpalekSzegedy11}  \footnote{Please check Section~\ref{prelims} for how the negative weight adversary bound composes}.  Therefore the
lower bound from one segment is multiplied by the usual quantum search cost for
finding one positive block among many.

In words, each inner segment costs \(\Omega(L_0)\), and the outer unique-OR over
about \(\ell/L_0\) segments contributes a square-root factor.  With the above
choice of \(L_0\), this gives
$\Omega\bigl((N\ell)^{1/4}\bigr)$
when \(N\leq \ell^3\).  When \(N\geq \ell^3\), the universe is large enough to
fit one cubic-scale segment of length \(\Theta(\ell)\), so the boundary lower
bound gives \(\Omega(\ell)\) directly.  Together these two regimes give the single linked-list lower bound
\[
\Omega\!\left(\min\{\ell,(N\ell)^{1/4}\}\right).
\]

\medskip
\emph{Double linked lists:}
For double linked lists, the outer strategy of serial composition is similar, but the  boundary
problem at the cubic threshold has to be modified.  The new issue is the predecessor oracle.  If we
used the same kind of boundary segment as in the single linked-list proof, then
a predecessor query at the public exit could reveal the last hidden vertex of
the segment.  If that last vertex is the possible marked vertex, then one query has learned exactly the information the adversary construction was trying to
hide.

To prevent this, the double linked-list boundary segment has a buffer.  Instead
of putting the possible marked vertex next to the public exit, we use a segment
of the form
$p \leftrightarrow x_1 \leftrightarrow \cdots \leftrightarrow x_L
  \leftrightarrow y_1 \leftrightarrow \cdots \leftrightarrow y_L
  \leftrightarrow q$.
Here \(p\) and \(q\) are public connector symbols.  The possible marked vertex is
\(x_L\), while the vertices \(y_1,\ldots,y_L\) form an unmarked buffer between
\(x_L\) and the public exit \(q\).  Later, this is formalized as the problem \(\mathrm{DBUFDEC}(N,L)\). 

 However,
the buffer creates a second problem.  An adversary matrix with weights assigned in the same way as in the singly linked-list case does not work. A single marking query at \(x_L\) increases the denominator a lot due to two paths splitting in the buffered part.

The final adversary matrix therefore weights the two halves in opposite directions.  On the \(X\)-part, \(x_1,\ldots,x_L\), it behaves like the single
linked-list adversary: pairs get larger weight when their first disagreement is
closer to the possible marked vertex \(x_L\).  On the \(Y\)-buffer, the weights
are reversed: the buffer is read from the public exit \(q\) back toward the
middle.  

This is the key difference for double linked lists.  As in the single linked-list case after these weights are modified, at the cubic scale the numerator is
large while every individual query contributes only a controlled denominator.
This gives the buffered boundary lower bound \(\Omega(L)\).

\subsection{Organization}

Section~\ref{prelims} talks about some quantum search primitives  used
throughout the paper. In Section~\ref{sec:singlelinkedlist} we give our bounds for the Single Linked List problem. Section~\ref{sec:uppbound} proves the upper bound and Section~\ref{sec:lower-bounds} proves the lower bound.  Section~\ref{app:doubly-linked-list} then gives bounds for the Double Linked List problem. Section~\ref{subsec:uppbound} and Section~\ref{subsec:doubly-lower-bound} give the upper and lower bounds respectively. 


\section{Preliminaries}
\label{prelims}

We use two standard amplitude-amplification primitives as black boxes.
The first allows exact preparation of the uniform superposition over a
known-size good set. The second is a fixed-point search routine that
works when only a lower bound on the good weight is known, thereby
avoiding the overshooting issue of ordinary Grover search.

\begin{lemma}
\label{lem:exact-good-state-prep}
Let $P:[N]\to\{0,1\}$ be a predicate, and let
 $G=\{x\in[N]:P(x)=1\}, t=|G|$.
Assume that $t\ge 1$ is known exactly and that $P$ can be evaluated reversibly with one query.
Then there is a quantum procedure using $O(\sqrt{N/t})$ queries to $P$ that prepares the state
$\frac{1}{\sqrt t}\sum_{x\in G}\ket{x}$.
\end{lemma}

\begin{proof}
Let $A\ket{0}=\frac{1}{\sqrt N}\sum_{x\in[N]}\ket{x}$.
Then the initial success probability is $a=t/N$. By exact amplitude amplification with known
success probability, one can obtain a good outcome with certainty using $O(1/\sqrt a)=
O(\sqrt{N/t})$ applications of $A$, $A^{-1}$, and the predicate oracle \cite[Sec.~2.1, Thm.~4]{BHMT02}.
Moreover, amplitude amplification preserves the two-dimensional subspace spanned by the initial
good and bad components, and the final exact step makes the bad component vanish. Hence the
final state is the normalized initial good component. For the uniform superposition on $[N]$, that
normalized good component is exactly $\frac{1}{\sqrt t}\sum_{x\in G}\ket{x}$.

\end{proof}

\begin{lemma}[\cite{YoderLowChuang14}]
\label{lem:fixed-point-search}
Let $A$ be a unitary with $A\ket{0} = \sum_x \alpha_x \ket{x}$,
and let $Q$ be a Boolean predicate on basis states that can be evaluated
reversibly. Suppose that for some known parameter $\varepsilon \in (0,1]$
one has the promise that either $\sum_{x:Q(x)=1} |\alpha_x|^2 = 0$
or $\sum_{x:Q(x)=1} |\alpha_x|^2 \ge \varepsilon$.
Then there is a quantum procedure using
$O(1/\sqrt{\varepsilon})$ calls to $A$, $A^\dagger$, and the reversible
evaluation of $Q$ such that, in the second case, it outputs a basis state
$x$ with $Q(x)=1$ with constant success probability. In the first case,
a final direct verification of the output candidate cleanly distinguishes
the no-case.
\end{lemma}

Let us make explicit how our two-oracle
notation fits the standard query model. Although an input is described as a
pair of oracles $(f,g)$, it may equivalently be viewed as a single
finite-alphabet input string
\[
X_{f,g}
=
\bigl( (f(u))_{u\in\{s\}\cup [N]},\,
       (g(v))_{v\in [N]} \bigr).
\]
The query coordinates of this single input form the disjoint union
\[
\mathcal I
=
\bigl(\{F\}\times(\{s\}\cup [N])\bigr)
\;\cup\;
\bigl(\{G\}\times [N]\bigr).
\]
A query to a coordinate of the form $(F,u)$ returns $f(u)$, while a query
to a coordinate of the form $(G,v)$ returns $g(v)$. Equivalently, in
reversible form, the corresponding single oracle acts as
$O_X\lvert F,u,z\rangle
=
\lvert F,u,z\oplus f(u)\rangle$
and
$O_X\lvert G,v,b\rangle
=
\lvert G,v,b\oplus g(v)\rangle $.
Thus the model is an ordinary finite-alphabet quantum query model whose
coordinates are the disjoint union of the $f$-coordinates and the
$g$-coordinates. Consequently, the usual negative-weight adversary bound
applies directly.

For the hidden double linked-list problem, the same convention applies with three
oracles instead of two. An input is equivalently a single finite-alphabet input
string
$X_{\operatorname{next}, \operatorname{prev},g}
=
\bigl((\operatorname{next}(u))_{u\in\{s\}\cup[N]},
      (\operatorname{prev}(u))_{u\in\{s\}\cup[N]},
      (g(v))_{v\in[N]}\bigr).$
The query coordinates are now the disjoint union
\[
\mathcal{I}_{\mathrm{dbl}}
=
\bigl(\{\operatorname{next}\}\times(\{s\}\cup[N])\bigr)
\cup
\bigl(\{\operatorname{prev}\}\times(\{s\}\cup[N])\bigr)
\cup
\bigl(\{G\}\times[N]\bigr).
\]
A query to a coordinate of the form \((\operatorname{next},u)\) returns \(\operatorname{next}(u)\), a query
to a coordinate of the form \((\operatorname{prev},u)\) returns \(\operatorname{prev}(u)\), and a query to a
coordinate of the form \((G,v)\) returns \(g(v)\). Thus the double linked-list
model is again an ordinary finite-alphabet quantum query model, now with three
disjoint families of query coordinates. Consequently, the negative-weight
adversary bound and its composition theorem apply directly in this model as
well.

To prove the lower bound, we use the negative-weight adversary bound $\mathrm{Adv}^{\pm}$, which is the standard general method for proving quantum query lower bounds for partial Boolean functions. 
 Let $ F:\Dom(F)\subseteq \Sigma^n\to\{0,1\}$
be a partial Boolean function.  An adversary matrix for \(F\) is a real matrix
\(\Gamma\) whose rows are indexed by \(1\)-inputs \(x\in F^{-1}(1)\) and whose
columns are indexed by \(0\)-inputs \(y\in F^{-1}(0)\).  For each query
coordinate \(i\in[n]\), let \(\Delta_i\) be the \(0\)-\(1\) matrix of the same
dimensions, defined by $ \Delta_i(x,y)=1
  \ \Longleftrightarrow\ 
  x_i\ne y_i .$
The negative-weight adversary bound~\cite{HLS07} is
\[
  \Advpm(F)
  =
  \max_{\Gamma\ne 0}
  \frac{\|\Gamma\|}
       {\max_i \|\Gamma\circ \Delta_i\|},
\]
where \(\|\cdot\|\) denotes spectral norm and \(\circ\) denotes entrywise
product. 

For our problem each input is a pair of oracles.  Let
\[
X=(f,g)
\qquad\text{and}\qquad
Y=(f',g')
\]
be two valid inputs.  A successor query at \(u\in\{s\}\cup[N]\)
distinguishes \(X\) from \(Y\) exactly when the two successor oracles give
different answers at \(u\).  Thus we define
$\Delta^f_u(X,Y)=1
\iff
f(u)\neq f'(u)$.
Similarly, a marking query at \(v\in[N]\) distinguishes \(X\) from \(Y\)
exactly when the two marking oracles give different answers at \(v\).  Thus
we define
$\Delta^g_v(X,Y)=1
\iff
g(v)\neq g'(v)$.
Therefore the denominator in the adversary bound is
\[
\max\left\{
\max_{u\in\{s\}\cup[N]}\|\Gamma\circ\Delta^f_u\|,
\;
\max_{v\in[N]}\|\Gamma\circ\Delta^g_v\|
\right\}.
\]

For a partial Boolean function $F$, the quantity $\mathrm{Adv}^{\pm}(F)$ measures the query hardness of $F$, and in fact
\[
Q(F)=\Theta(\mathrm{Adv}^{\pm}(F))
\]
for bounded-error quantum query complexity~\cite[Theorem~1.1]{LeeMittalReichardtSpalekSzegedy11}.

We will also use the fact that the negative weight adversary bound composes.  Let
\[
F:\operatorname{Dom}(F)\subseteq \{0,1\}^m \to \Lambda
\qquad\text{and}\qquad
G:\operatorname{Dom}(G)\subseteq \Sigma^r \to \{0,1\}
\]
be partial functions, where $\Lambda,\Sigma$ are finite alphabets. The block composition $F\circ G^m$ is the partial function on $m$ pairwise disjoint input blocks $x^{(1)},\dots,x^{(m)}\in \operatorname{Dom}(G)$ defined by
\[
(F\circ G^m)(x^{(1)},\dots,x^{(m)})
=
F\bigl(G(x^{(1)}),\dots,G(x^{(m)})\bigr),
\]
whenever the output tuple $\bigl(G(x^{(1)}),\dots,G(x^{(m)})\bigr)$ lies in $\operatorname{Dom}(F)$. Thus each copy of $G$ acts on its own disjoint block, and the outer function $F$ sees only the resulting $m$ output bits. In the Boolean-output setting relevant here, one has
\[
\mathrm{Adv}^{\pm}(F\circ G^m) \ge \mathrm{Adv}^{\pm}(F)\,\mathrm{Adv}^{\pm}(G)
\]
by the adversary composition theorem; see~\cite[Lemma~5.2]{LeeMittalReichardtSpalekSzegedy11}.

Although we formulate our lower bounds using the general (negative-weight) adversary bound $\operatorname{Adv}^{\pm}$, the explicit adversary matrices
constructed are in fact entrywise nonnegative. 
We nevertheless use the general adversary framework because our oracle
problems are naturally partial finite-alphabet query problems. The composition theorem for the general adversary bound
that we use above~\cite[Lemma~5.2]{LeeMittalReichardtSpalekSzegedy11} applies directly in this setting. By comparison, the standard composition theorem for the
nonnegative adversary bound in~\cite{HLS07} seems to be stated for
Boolean-coordinate inputs.

\section{Single Linked List} \label{sec:singlelinkedlist}

\subsection{Upper Bound} \label{sec:uppbound}

We now prove the upper bound for the hidden linked-list problem. We begin with the search problem $\HLLSearch(N,\ell)$. The same upper bound for $\HLLDEC(N,\ell)$ then follows immediately. The case $\ell=1$ is trivial: one queries $f(s)$ to obtain the unique list vertex and then queries $g$ on that vertex. Thus throughout this section we may assume $\ell\ge 2$.

Call a vertex $x\in [N]$ \emph{live} if $f(x)\neq \pp$. Define a predicate 
$\Live(x) =[f(x)\neq \pp]$ for live vertices and denote by $L=\{x\in [N]:\Live(x)=1\}$ a set of live vertices.
By the input promise, the live vertices are exactly the nonterminal list vertices.

We invoke the following subroutine obtained by amplitude amplification (Lemma~\ref{lem:exact-good-state-prep}) in the search algorithm to generate a uniform superposition over live vertices.
\begin{lemma}\label{lem:live-prep}
Assume $\ell \ge 2$, and define $L = \{x \in [N] : f(x) \neq \pp\}$.
Then $L = \{a_1,\dots,a_{\ell-1}\}$ and $|L|=\ell-1$. Moreover, there is a unitary $P_L$ such that
$P_L\ket{0} = \ket{L} = \frac{1}{\sqrt{\ell-1}}\sum_{x\in L} \ket{x},$
and $P_L$ uses $O(\sqrt{N/\ell})$ queries to $f$.
\end{lemma}

\begin{proof}
For each $1\le i<\ell$, we have $f(a_i)=a_{i+1}\neq \pp$, so $a_i\in L$. On the other hand, $f(a_\ell)=\pp$, so $a_\ell\notin L$. Finally, if $x\notin \{a_1,\dots,a_\ell\}$, then by promise $f(x)=\pp$, so $x\notin L$. Hence $L=\{a_1,\dots,a_{\ell-1}\},$ and therefore $|L|=\ell-1$.

Now consider the predicate $\Live(x)$. This predicate can be evaluated with one query to $f$, and the number of good inputs is known exactly, namely $\ell-1$. Starting from the uniform superposition on $[N]$, Lemma~\ref{lem:exact-good-state-prep} gives an exact amplitude-amplification procedure that prepares the normalized good component using
$O\!\left(\sqrt{\frac{N}{\ell-1}}\right)=O\!\left(\sqrt{\frac{N}{\ell}}\right)$
queries. Since every good basis state starts with the same amplitude in the uniform superposition on $[N]$, the resulting state is precisely
$
\ket{L}=\frac{1}{\sqrt{\ell-1}}\sum_{x\in L}\ket{x}.
$
This proves the lemma.
\end{proof}

\subsubsection*{The algorithm}

 Fix a parameter \(k \in \{1,\ldots,\ell-1\}\).     The algorithm first inspects the first $k$ vertices of the hidden list explicitly. If the marked vertex is not found there, it then searches for a \emph{live} vertex from which a short forward walk reaches the marked vertex.

To formalize this, define for $x\in L$ the predicate
$$
\Good_k(x)=1
\iff
\text{the walk }x,f(x),f^{(2)}(x),\dots,f^{(k)}(x)\text{ encounters a marked vertex.}
$$
Thus $\Good_k(x)=1$ precisely when starting from $x$ and walking forward for at most $k$ successor steps reaches the marked vertex.

The following pseudocode gives the procedure.

\begin{center}
\fbox{
\begin{minipage}{0.93\linewidth}
\small
\noindent \textbf{Pseudocode for $\HLLSearch(N,\ell)$ with parameter $k$}

\medskip

\noindent \textbf{Input:} Oracle access to \(f\) and \(g\) and a parameter \(k \in \{1,\ldots,\ell-1\}\).

\begin{enumerate}
    \item Starting from $s$, walk forward through the first $k$ list vertices and query $g$ on each one.
    \item If a marked vertex is found during this prefix scan, output it and halt.
    \item Prepare the state $\ket{L}$ using Lemma~\ref{lem:live-prep}.
    \item Run the fixed-point search procedure of Lemma~\ref{lem:fixed-point-search} over the support $L$ with predicate $\Good_k$.
    \item If the search returns a candidate $x\in L$, walk forward from $x$ for at most $k$ successor steps, querying $g$ along the way.
    \item If a marked vertex is found during this recovery walk, output it; otherwise output failure.
\end{enumerate}
\end{minipage}
}
\end{center}

We now justify a key ingredient used by the pseudocode: that $\Good_k$ can be checked efficiently.

\begin{lemma}[Checking $\Good_k$]\label{lem:good-check}
The predicate $\Good_k(x)$ can be evaluated reversibly using $O(k)$ queries to $(f,g)$.
\end{lemma}

\begin{proof}
Given $x$, compute the forward walk $x,\ f(x),\ f^{(2)}(x),\ \dots,\ f^{(k)}(x),$
querying $g$ on each visited vertex and recording in an ancilla whether a marked vertex has been seen. If the walk reaches $\pp$ before $k$ steps, continue instead in a fixed dummy sink state so that the computation still has exactly $k$ steps. After writing the value of $\Good_k(x)$ into a target qubit, reverse the entire computation to uncompute all workspace. The total query cost is $O(k)$.
\end{proof}

\subsubsection*{Correctness}

The reason for scanning the first $k$ vertices explicitly is that, if the marked vertex lies deeper in the list, then many \emph{live} vertices become good starting points.

\begin{lemma}\label{lem:good-fraction}
Assume the first $k$ list vertices are unmarked and the unique marked vertex is $a_t$ for some $t>k$. Then among the live vertices $L=\{a_1,\dots,a_{\ell-1}\}$, the number of good vertices is either $k$ or $k+1$. In particular,
$|\{x\in L : \Good_k(x)=1\}| \ge k.$
\end{lemma}

\begin{proof}
Fix a \emph{live} vertex $a_j$, where $1\le j\le \ell-1$. By definition, $\Good_k(a_j)=1$ exactly when the marked vertex $a_t$ appears among $a_j,\ a_{j+1},\ \dots,\ a_{j+k}.$
Equivalently, $0\le t-j\le k.$
Thus $a_j$ is good exactly when $t-k\le j\le t.$
Since $j$ must also satisfy $1\le j\le \ell-1$, the good live indices are precisely $\{j: t-k\le j\le \min\{t,\ell-1\}\}.$

Because $t>k$, the lower endpoint $t-k$ is at least $1$. There are now two cases.

If $t=\ell$, then the upper endpoint is $\ell-1$, so the number of good indices is
$(\ell-1)-(t-k)+1 = k.$

If $t\le \ell-1$, then the upper endpoint is $t$, so the number of good indices is
$t-(t-k)+1 = k+1.$

Hence the number of good \emph{live} vertices is either $k$ or $k+1$, and in particular it is at least $k$.
\end{proof}

We now prove that the pseudocode above correctly solves $\HLLSearch(N,\ell)$.

If the marked vertex lies among the first $k$ list vertices, then the initial prefix scan finds it directly, so the algorithm is correct in this case. Now suppose that the first $k$ list vertices are unmarked.

First consider the no-instance. If there is no marked vertex anywhere on the list, then by definition $\Good_k(x)=0$ for every \emph{live} vertex $x\in L$. Therefore the fixed-point search procedure is in its no-case. By Lemma~\ref{lem:fixed-point-search}, the algorithm correctly outputs failure.

Next consider the yes-instance, and let the unique marked vertex be $a_t$ with $t>k$. By Lemma~\ref{lem:good-fraction}, at least $k$ of the $\ell-1$ \emph{live} vertices satisfy $\Good_k$. Therefore, in the state
$\ket{L}=\frac{1}{\sqrt{\ell-1}}\sum_{x\in L}\ket{x},$
the total good weight is at least
$\frac{k}{\ell-1}.$

Applying Lemma~\ref{lem:fixed-point-search} with state-preparation unitary $P_L$, predicate $\Good_k$, and
$\varepsilon = \frac{k}{\ell-1},$
we obtain, with constant success probability, a candidate $x\in L$ such that $\Good_k(x)=1$.

By the definition of $\Good_k$, once such an $x$ has been found, a forward walk from $x$ of length at most $k$ must encounter the marked vertex. Therefore the final recovery walk indeed finds and outputs the marked vertex.

We have thus shown that the algorithm outputs the unique marked vertex when one exists, and outputs failure otherwise. In particular, it solves $\HLLSearch(N,\ell)$ with bounded error. Since the decision problem is obtained simply by checking whether the search algorithm found a marked vertex, the same argument yields an upper bound for $\HLLDEC(N,\ell)$ as well.

\subsubsection*{Query complexity}

We now analyze the number of oracle queries used by the algorithm.

The initial prefix scan costs $O(k)$ queries. If the scan does not find a marked vertex, we prepare $\ket{L}$ using Lemma~\ref{lem:live-prep}, which costs $O(\sqrt{N/\ell})$ queries to $f$.

By Lemma~\ref{lem:good-check}, one evaluation of $\Good_k$ costs $O(k)$ queries.

By Lemma~\ref{lem:good-fraction}, in the yes-case the fraction of good basis states inside $\ket{L}$ is at least $\frac{k}{\ell-1},$
while in the no-case it is $0$. Hence Lemma~\ref{lem:fixed-point-search} applies with
$\varepsilon=\frac{k}{\ell-1},$
and uses
$O\!\left(\sqrt{\frac{\ell-1}{k}}\right)=O\!\left(\sqrt{\frac{\ell}{k}}\right)$
outer iterations.

Each outer iteration uses one call to $P_L$, one call to $P_L^\dagger$, and one evaluation of $\Good_k$. Finally, after a good starting point has been found, the recovery walk costs an additional $O(k)$ queries. Therefore the total query complexity is
$O\!\left(k + \left(\sqrt{\frac{N}{\ell}} + k\right)\sqrt{\frac{\ell}{k}}\right).$

The additive $O(k)$ term is in fact dominated by the second term, because $k\le \ell$ implies
$\left(\sqrt{\frac{N}{\ell}} + k\right)\sqrt{\frac{\ell}{k}}
\ge k\sqrt{\frac{\ell}{k}}
= \sqrt{k\ell} \ge k.$ 

 We now choose \(k\). If \(N \le \ell^3\), set
$k = \min\left\{\ell-1,\left\lceil \sqrt{\frac{N}{\ell}} \right\rceil\right\}$.
Then \(1 \le k \le \ell-1\), and $\varepsilon = \frac{k}{\ell-1}$
used in Lemma~\ref{lem:fixed-point-search} satisfies \(\varepsilon \in (0,1]\). Moreover, since
\(N \ge \ell\) and \(N \le \ell^3\), we have
$1 \le \sqrt{\frac{N}{\ell}} \le \ell,$
and this choice still gives
$k = \Theta\left(\sqrt{\frac{N}{\ell}}\right).$
Thus the two terms \(\sqrt{N/\ell}\) and \(k\) are equal up to constants.

Substituting this choice gives
$ O\!\left(
k + \left(\sqrt{\frac{N}{\ell}} + k\right)\sqrt{\frac{\ell}{k}}
\right)
=
O\!\left(
k\sqrt{\frac{\ell}{k}}
\right)
=
O(\sqrt{k\ell})
=
O\!\bigl((N\ell)^{1/4}\bigr).$

If $N>\ell^3$, then we simply traverse the hidden list from the head and query $g$ on every visited vertex. This uses $O(\ell)$ queries.

We have therefore proved the following theorem,

\begin{theorem}[Upper bound]\label{thm:upper}
For all $N \ge \ell \ge 1$,
$$
Q(\HLLDEC(N,\ell)),\; Q(\HLLSearch(N,\ell))
  \in O\!\left(\min\{\ell,(N\ell)^{1/4}\}\right).
$$
\end{theorem}

\subsection{Lower Bound}
\label{sec:lower-bounds}

We now prove the matching lower bound
\[
  Q(\HLLDEC(N,\ell)),\ Q(\HLLSearch(N,\ell))
  \in
  \Omega\!\left(\min\{\ell,(N\ell)^{1/4}\}\right).
\]

The lower bound essentially has two ingredients.  The first is a special case in which the ambient universe
has the ``critical'' size \(L^3\) for a list segment of length \(L\).  In that
case we will prove, by an explicit adversary matrix derived from the negative adversary lower bound~\cite{HLS07, LeeMittalReichardtSpalekSzegedy11}, that even a quantum
algorithm needs \(\Omega(L)\) queries.  

Once this critical case is known, the rest of the lower bound is a packing
argument.  Think of the critical instance as a hard segment: it is a hidden segment
with a known entrance and a known exit.  Because the entrance and exit are
public, we can place many such segments in a row, joined by fixed unmarked bridge
vertices, and obtain one longer hidden linked list.  If we promise that at most
one segment contains a marked vertex, then deciding whether the whole list contains
a mark is exactly the problem of taking the OR of the answers of the segments,
under the unique-solution promise.

This is where the parameter \((N\ell)^{1/4}\) comes from.  Suppose we pack
\(m\) hard segments, each of length \(L_0\), with each segment living at its own cubic
threshold ambient space \(N_0=L_0^3\).  The boundary lower bound gives a cost \(\Omega(L_0)\)
inside one segment.  The outer unique-OR over \(m\) segments contributes the usual
\(\Omega(\sqrt m)\) quantum search cost.  By the adversary composition theorem,
the packed instance therefore has lower bound $\Omega(\sqrt m\,L_0)$.
Choosing \(L_0\) and \(m\) so that the packed construction fits inside the
available \((N,\ell)\)-instance will make this quantity equal, up to constants,
to \((N\ell)^{1/4}\) in the regime \(N\le \ell^3\).  In the complementary
regime \(N\ge \ell^3\), a single critical segment of length \(\Theta(\ell)\)
already fits, giving the linear lower bound \(\Omega(\ell)\).

Thus the proof is organized as follows.  We first define the segment, which we call
an anchored segment.  We then prove, assuming the critical lower bound for one
anchored segment, that the desired lower bound follows for all \(N\) and
\(\ell\).  Only after this reduction is complete do we return to the critical
case and prove the anchored-segment lower bound of $\Omega(L)$ by constructing an adversary
matrix.

\subsubsection{Anchored segments}
\label{subsec:anchored-segments}

The hidden linked-list problem has a public start symbol \(s\), but the last
vertex simply points to \(\bot\).  For composition it is more convenient to use
segments with both a public entry point and a public exit point.  The exit point
will later serve as the bridge to the next segment.

\begin{definition}[Anchored segment decision problem]
\label{def:asegdec}
Fix integers \(N\ge L\ge 1\).  An input to
\(\ASEGDEC(N,L)\) consists of oracle access to
\[
  f:\{u_{\mathrm{in}}\}\cup [N]\to [N]\cup\{u_{\mathrm{out}},\bot\}
\]
and
\[
  g:[N]\to\{0,1\},
\]
with the promise that there are pairwise distinct vertices
\(x_1,\ldots,x_L\in [N]\) such that
\[
  f(u_{\mathrm{in}})=x_1,\qquad
  f(x_i)=x_{i+1}\quad (1\le i<L),\qquad
  f(x_L)=u_{\mathrm{out}},
\]
and \(f(y)=\bot\) for every
\(y\in [N]\setminus\{x_1,\ldots,x_L\}\).
The marking promise is
\[
  \sum_{i=1}^L g(x_i)\le 1,
  \qquad
  g(y)=0\quad\text{off the segment}.
\]
The output is \(1\) if and only if some segment vertex is marked.

We use the same query convention as in Definition~\ref{def:def1}: one query means one reversible
successor query to \(f\) or one reversible marking query to \(g\), with the range of \(f\)
encoded using a fixed binary encoding of \([N]\cup\{u_{\mathrm{out}},\bot\}\).  

\end{definition}

The lower bound for anchored segments at the cubic threshold is the core technical statement of the section.

\begin{theorem}[Boundary lower bound]
\label{thm:boundary-aseg}
There is an absolute constant \(c_{\mathrm{bd}}>0\) such that, for every
integer \(L\ge 1\),
\[
  \Advpm(\ASEGDEC(L^3,L))\ge c_{\mathrm{bd}} L .
\]
\end{theorem}

We defer the proof of Theorem~\ref{thm:boundary-aseg} to
Section~\ref{subsec:boundary-proof}.  For now, we explain why this boundary
case implies the full lower bound.

\subsubsection{Two standard adversary facts}
\label{subsec:adversary-standard-facts}

We shall use two standard properties of the negative-weight adversary bound.

First, the adversary bound is monotone under restrictions.  If a problem \(F\)
has a restriction that is isomorphic to another problem \(G\), then
\[
  \Advpm(F)\ge \Advpm(G).
\]
Indeed, any adversary matrix for \(G\) can be viewed as an adversary matrix for
\(F\) by placing it on the restricted set of inputs and filling all other rows
and columns with zeros.

Second, we use the repeated-block form of the adversary composition theorem.
Let
  $F:\Dom(F)\subseteq\{0,1\}^m\to\Lambda$
be a partial outer function, and let $G:\Dom(G)\subseteq\Sigma^r\to\{0,1\}$
be a partial Boolean inner function.  The block composition \(F\circ G^m\)
takes \(m\) disjoint inputs
  $y^{(1)},\ldots,y^{(m)}\in \Dom(G)$
and returns
\[
  (F\circ G^m)(y^{(1)},\ldots,y^{(m)})
  =
  F(G(y^{(1)}),\ldots,G(y^{(m)})),
\]
whenever the tuple of inner outputs lies in \(\Dom(F)\).  The adversary
composition theorem gives
\[
  \Advpm(F\circ G^m)
  \ge
  \Advpm(F)\Advpm(G).
\]
We shall apply this only when \(F\) is unique-OR.

Let \(\UOR_m\) denote OR on \(m\) bits under the promise that at most one
input bit is equal to \(1\).  This is unique-solution unstructured search.
By Grover's algorithm~\cite{Grover97} and the standard quantum lower bound for
unstructured search~\cite{BBBV97},
  $Q(\UOR_m)=\Theta(\sqrt m)$.
Using the tight characterization \(Q(F)=\Theta(\Advpm(F))\)
from~\cite{LeeMittalReichardtSpalekSzegedy11}, we obtain
\[
  \Advpm(\UOR_m)=\Theta(\sqrt m).
\]

\subsubsection{Serial packing}
\label{subsec:serial-packing}

The next lemma explains how to place several anchored segments in series inside
one hidden linked list.  This is the mechanism that converts the boundary lower
bound into a smaller universe lower bound.

\begin{lemma}[Serial packing]
\label{lem:serial-packing}
Let \(m,N_0,L_0\ge 1\) with \(N_0\ge L_0\). There is a restriction of
\[
  \HLLDEC(mN_0+m,\ mL_0+m)
\]
to inputs obtained by concatenating \(m\) anchored segments of parameters
\((N_0,L_0)\) in series.  Under the promise that at most one segment is
positive, this restricted problem is exactly
\[
  \UOR_m\circ \ASEGDEC(N_0,L_0)^m .
\]
\end{lemma}

\begin{proof}
Choose pairwise disjoint address sets
\[
  U_1,\ldots,U_m,
  |U_t|=N_0, t\in [m],
\]
and choose distinct bridge vertices
  $b_1,\ldots,b_m$
outside \(U_1\cup\cdots\cup U_m\).  For convenience put \(b_0=s\), where \(s\)
is the public start symbol of the hidden linked-list problem.

For each \(t\in[m]\), take an input to \(\ASEGDEC(N_0,L_0)\) whose private
address universe is \(U_t\), whose entry symbol is \(b_{t-1}\), and whose exit
symbol is \(b_t\).  If the hidden segment in block \(t\) is
\[
  x_{t,1}\to x_{t,2}\to\cdots\to x_{t,L_0},
  \qquad x_{t,i}\in U_t,
\]
then in the global linked list we set
\[
  f(b_{t-1})=x_{t,1},
\]
\[
  f(x_{t,i})=x_{t,i+1}\quad (1\le i<L_0),
\]
and
\[
  f(x_{t,L_0})=b_t.
\]
Finally set
\[
  f(b_m)=\bot.
\]
All addresses not used by these blocks and bridges are made dead by setting
their successor value to \(\bot\).

The resulting hidden list is
\[
  s=b_0
  \to x_{1,1}\to\cdots\to x_{1,L_0}\to b_1
  \to x_{2,1}\to\cdots\to x_{2,L_0}\to b_2
  \to\cdots
  \to x_{m,1}\to\cdots\to x_{m,L_0}\to b_m
  \to\bot .
\]
Thus the list has exactly
\[
  mL_0+m
\]
vertices, and the ambient universe used has exactly
\[
  mN_0+m
\]
addresses.

The marking oracle is defined blockwise: on \(U_t\) it is the marking oracle of
the \(t\)-th anchored segment, and every bridge vertex is forced to be unmarked:
\[
  g(b_t)=0
  \qquad (1\le t\le m).
\]

Now encode each local anchored-segment input as a separate block
\[
  y^{(t)}\in\Dom(\ASEGDEC(N_0,L_0)).
\]
The only variable coordinates of the global input are the coordinates inside
the \(m\) disjoint local blocks.  The bridge data are fixed and carry no input
information.

The global hidden list contains a marked vertex if and only if at least one
local anchored segment is positive.  Under the promise that at most one local
segment is positive, the global decision problem is therefore precisely
\[
  \UOR_m\circ \ASEGDEC(N_0,L_0)^m .
\]
This proves the lemma.
\end{proof}

\subsubsection{From the boundary lower bound to all parameters}
\label{subsec:one-shot-amplification}

We now prove the full adversary lower bound for \(\HLLDEC(N,\ell)\), assuming
Theorem~\ref{thm:boundary-aseg}.

\begin{theorem}[Lower bound for hidden linked-list decision]
\label{thm:hll-lower-decision}
There is an absolute constant \(c>0\) such that, for all integers
\(N\ge \ell\ge 1\),
\[
  \Advpm(\HLLDEC(N,\ell))
  \ge
  c\min\{\ell,(N\ell)^{1/4}\}.
\]
\end{theorem}

\begin{proof}
For \(\ell\le 3\), the desired right-hand side is \(O(1)\).  The problem has a
constant lower bound even on the restriction where the hidden path is fixed and
only one marked bit is variable.  Thus, after decreasing the constant \(c\), we
may assume throughout the rest of the proof that
  $\ell\ge 4$.

We split into two regimes.

\paragraph{Regime I: \(N\ge \ell^3\).}

Here the desired lower bound is linear in \(\ell\), because
\[
  N\ge \ell^3
  \quad\Longrightarrow\quad
  (N\ell)^{1/4}\ge \ell,
\]
and hence
\[
  \min\{\ell,(N\ell)^{1/4}\}=\ell.
\]

We realize one anchored segment of length \(\ell-1\) inside the hidden linked
list and use one bridge vertex to terminate it.  Set
\[
  m=1,\qquad
  L_0=\ell-1,\qquad
  N_0=(\ell-1)^3.
\]
Then
\[
  mL_0+m=(\ell-1)+1=\ell,
\]
and
\[
  mN_0+m=(\ell-1)^3+1\le \ell^3\le N.
\]
After padding all remaining addresses with dead vertices, this gives a
restriction of \(\HLLDEC(N,\ell)\) whose nontrivial part is
\[
  \ASEGDEC((\ell-1)^3,\ell-1).
\]
By monotonicity under restriction and Theorem~\ref{thm:boundary-aseg},
\[
  \Advpm(\HLLDEC(N,\ell))
  \ge
  \Advpm(\ASEGDEC((\ell-1)^3,\ell-1))
  \ge
  c_{\mathrm{bd}}(\ell-1)
  =
  \Omega(\ell).
\]
This proves the claimed lower bound in the regime \(N\ge \ell^3\).

\paragraph{Regime II: \(N < \ell^3\).}

This is the regime where the lower bound should be
\[
  \Omega((N\ell)^{1/4}).
\]
The idea is to pack many smaller cubic-threshold anchored segments in series.

Set
\[
  \alpha=\sqrt{N/\ell}.
\]
Since \(N\ge \ell\), we have \(\alpha\ge 1\).  Since \(N\le \ell^3\), we have
\(\alpha\le \ell\).

Choose
\[
  L_0=\max\{1,\lfloor \alpha/8\rfloor\},
  \qquad
  N_0=L_0^3,
\]
and
\[
  m=\max\left\{1,
      \left\lfloor \frac{\ell}{2(L_0+1)}\right\rfloor
    \right\}.
\]
Thus each inner block lies exactly at the cubic threshold: $N_0=L_0^3.$

 We serially pack \(m\) anchored
segments of parameters \((N_0,L_0)\), using one unmarked bridge vertex after each segment, as in Lemma~\ref{lem:serial-packing}. This creates a list length $mL_0+m=m(L_0+1),$ and uses $mN_0+m=m(N_0+1)$ addresses. If this list length is smaller than \(\ell\), we append a fixed
unmarked tail so that the total hidden list has length exactly \(\ell\). All remaining addresses are set to be dead.

Thus, to obtain a valid restriction of \(\HLLDEC(N,\ell)\), it is enough to
verify the following three feasibility conditions:
\begin{enumerate}
  \item The serially packed composition of anchored segments has list length at most
  \(\ell\), namely $m(L_0+1)\le \ell.$

  \item The serially packed composition of anchored segments uses at most \(N\) addresses, namely $m(N_0+1)\le N.$

  \item After the serially packed composition of anchored segments is positioned, there is enough unused address space
  to append the fixed unmarked tail needed to bring the total list length to  exactly \(\ell\), namely $ N-(mN_0+m)\ge \ell-(mL_0+m).$
\end{enumerate}

We verify these three conditions in turn next.

\medskip
\noindent
\textbf{Condition 1: serial composition of anchored segment has length at most \(\ell\).}

We claim that
\[
  m(L_0+1)\le \ell.
\]
If \(m>1\), then by definition
\[
  m
  =
  \left\lfloor \frac{\ell}{2(L_0+1)}\right\rfloor
  \le
  \frac{\ell}{2(L_0+1)},
\]
so
\[
  m(L_0+1)\le \ell/2\le \ell.
\]
If \(m=1\), then we need \(L_0+1\le \ell\).  Since
\[
  L_0\le \max\{1,\alpha/8\}\le \max\{1,\ell/8\},
\]
and \(\ell\ge 4\), this gives \(L_0+1\le \ell\).  Hence
\[
  mL_0+m=m(L_0+1)\le \ell
\]
in all cases.

\medskip
\noindent
\textbf{Condition 2: serial composition of anchored segments uses at most \(N\) addresses.}

We claim that
\[
  m(N_0+1)=m(L_0^3+1)\le N.
\]

First suppose \(L_0=1\).  Then \(N_0+1=2\), and
\[
  m
  =
  \max\left\{1,\left\lfloor\frac{\ell}{4}\right\rfloor\right\}.
\]
Since \(\ell\ge 4\),
\[
  2m\le \ell\le N.
\]
Thus the ambient budget holds when \(L_0=1\).

Now suppose \(L_0>1\). Then \(L_0=\lfloor \alpha/8\rfloor\), so \(L_0\le \alpha/8\). Moreover, since \(L_0>1\), we have \(L_0\ge 2\). Also \(\alpha\le \ell\), and hence \(L_0\le \ell/8\). Therefore
\[
  \frac{\ell}{2(L_0+1)}
  \ge
  \frac{8L_0}{2(L_0+1)}
  =
  \frac{4L_0}{L_0+1}
  >2.
\]
Thus,
\[
  m=
  \left\lfloor \frac{\ell}{2(L_0+1)}\right\rfloor
  \le
  \frac{\ell}{2(L_0+1)}
  \le
  \frac{\ell}{2L_0}.
\]
Since \(L_0\ge 2\), we also have \(L_0^3+1\le 2L_0^3\). Hence
\[
  m(N_0+1)
  =
  m(L_0^3+1)
  \le
  \frac{\ell}{2L_0}\cdot 2L_0^3
  =
  \ell L_0^2
  \le
  \ell(\alpha/8)^2
  =
  \frac{N}{64}
  \le N.
\]
So the packed blocks and bridges fit inside the ambient universe.

\medskip
\noindent
\textbf{Condition 3: after the serial composition of anchored segments there is enough unused space to append the fixed tail.}

After serially packing the \(m\) anchored segments and the \(m\) bridges, the
current list length is
\[
  mL_0+m.
\]
If this is smaller than \(\ell\), we append a fixed unmarked tail of length
\[
  \ell-(mL_0+m).
\]
This tail uses one fresh address per added list vertex.  Thus we need
\[
  N-(mN_0+m)
  \ge
  \ell-(mL_0+m).
\]
Equivalently,
\[
  N-\ell
  \ge
  m(N_0-L_0).
\]

If \(L_0=1\), then \(N_0-L_0=1^3-1=0\), so this condition is immediate.

Now suppose \(L_0>1\).  As above, \(L_0=\lfloor \alpha/8\rfloor\ge 2\), and
therefore \(\alpha\ge 16\).  Also
\[
  m\le \frac{\ell}{2L_0}.
\]
Hence
\[
  m(N_0-L_0)
  \le
  mL_0^3
  \le
  \frac{\ell}{2L_0}\cdot L_0^3
  =
  \frac{\ell L_0^2}{2}
  \le
  \frac{\ell(\alpha/8)^2}{2}
  =
  \frac{N}{128}.
\]
On the other hand,
\[
  N-\ell
  =
  \ell(\alpha^2-1)
  =
  N\left(1-\frac1{\alpha^2}\right).
\]
Since \(\alpha\ge 16\),
\[
  1-\frac1{\alpha^2}
  \ge
  1-\frac1{16^2}
  =
  \frac{255}{256}.
\]
Thus
\[
  N-\ell\ge \frac{255}{256}N>\frac{N}{128}
  \ge
  m(N_0-L_0).
\]
So the appended fixed tail also fits.

\medskip

The three conditions show that the serially
packed \(m\) anchored segments of parameters \((N_0,L_0)\) give a valid
restriction of \(\HLLDEC(N,\ell)\).  By Lemma~\ref{lem:serial-packing}, under
the promise that at most one block is positive, this
restriction is exactly
\[
  \UOR_m\circ \ASEGDEC(N_0,L_0)^m .
\]

Therefore, by monotonicity under restriction and the adversary composition
theorem,
\[
  \Advpm(\HLLDEC(N,\ell))
  \ge
  \Advpm(\UOR_m)\,
  \Advpm(\ASEGDEC(N_0,L_0)).
\]
Since \(N_0=L_0^3\), Theorem~\ref{thm:boundary-aseg} gives
\[
  \Advpm(\ASEGDEC(N_0,L_0))
  \ge
  c_{\mathrm{bd}}L_0.
\]
Also,
\[
  \Advpm(\UOR_m)=\Theta(\sqrt m).
\]
Hence
\[
  \Advpm(\HLLDEC(N,\ell))
  \ge
  \Omega(\sqrt m\,L_0).
\]

It remains only to compare \(\sqrt m\,L_0\) with \((N\ell)^{1/4}\).

Let
\[
  x=\frac{\ell}{2(L_0+1)}.
\]
Since
\[
  m=\max\{1,\lfloor x\rfloor\},
\]
we always have
\[
  m\ge x/2.
\]
Indeed, if \(x<2\), then \(m\ge 1>x/2\); while if \(x\ge 2\), then
\(\lfloor x\rfloor\ge x/2\).  Therefore
\[
  m\ge \frac{\ell}{4(L_0+1)}.
\]
It follows that
\[
  \sqrt m\,L_0
  \ge
  \frac12 L_0\sqrt{\frac{\ell}{L_0+1}}
  =
  \frac12\sqrt{\ell\cdot\frac{L_0^2}{L_0+1}}.
\]
Since \(L_0\ge 1\),
\[
  \frac{L_0^2}{L_0+1}
  \ge
  \frac{L_0}{2}.
\]
Thus
\[
  \sqrt m\,L_0
  \ge
  \frac{1}{2\sqrt 2}\sqrt{\ell L_0}.
\]

Now recall that
\[
  \alpha=\sqrt{N/\ell},
\]
so
\[
  (N\ell)^{1/4}
  =
  \sqrt{\alpha\ell}.
\]

If \(\alpha<8\), then \(L_0=1\), and
\[
  (N\ell)^{1/4}
  =
  \sqrt{\alpha\ell}
  \le
  \sqrt{8\ell}
  =
  O(\sqrt{\ell L_0}).
\]
If \(\alpha\ge 8\), then
\[
  L_0=\max\{1,\lfloor \alpha/8\rfloor\}\ge \alpha/16,
\]
and hence
\[
  \sqrt{\ell L_0}
  \ge
  \frac14\sqrt{\alpha\ell}
  =
  \Omega((N\ell)^{1/4}).
\]
In both cases,
\[
  \sqrt m\,L_0
  =
  \Omega((N\ell)^{1/4}).
\]
Therefore
\[
  \Advpm(\HLLDEC(N,\ell))
  \ge
  \Omega((N\ell)^{1/4})
\]
throughout the regime \(N\le \ell^3\).

All constants above are absolute.  Taking \(c>0\) to be the minimum of the
constant obtained from the finitely many cases \(\ell\le 3\), the constant in
Regime I, and the constant in Regime II, the three estimates combine to give
\[
  \Advpm(\HLLDEC(N,\ell))
  \ge c\min\{\ell,(N\ell)^{1/4}\}.
\]
This proves the theorem.
\end{proof}

Finally, the search lower bound follows immediately.

\begin{theorem}[Lower bound for decision and search]
\label{cor:hll-lower-final}
For all integers \(N\ge \ell\ge 1\),
\[
  Q(\HLLDEC(N,\ell)),\ Q(\HLLSearch(N,\ell))
  \in
  \Omega\!\left(\min\{\ell,(N\ell)^{1/4}\}\right).
\]
\end{theorem}

\begin{proof}
By the tight characterization of bounded-error quantum query complexity by the
negative-weight adversary bound,
\[
  Q(\HLLDEC(N,\ell))
  =
  \Omega(\Advpm(\HLLDEC(N,\ell))).
\]
Theorem~\ref{thm:hll-lower-decision} therefore gives
\[
  Q(\HLLDEC(N,\ell))
  \in
  \Omega\!\left(\min\{\ell,(N\ell)^{1/4}\}\right).
\]
The search problem is at least as hard as the decision problem: any algorithm
that outputs the marked vertex, or reports that no marked vertex exists, also
decides whether a mark exists.  Hence the same lower bound holds for
\(\HLLSearch(N,\ell)\).
\end{proof}

\subsubsection{The boundary lower bound}
\label{subsec:boundary-proof}

It remains to prove Theorem~\ref{thm:boundary-aseg}.   The rest of
this section is devoted to that construction.  The proof is an application of
the negative-weight adversary method~\cite{HLS07}. Recall the definition of the adversary matrix from Section~\ref{prelims}.  The main idea is to restrict the anchored
segment problem to a very symmetric family of inputs.  In this family the hidden
path is forced to choose exactly one vertex from each of \(r\) layers.  Thus the hidden path is specified by a tuple
\[
z=(z_1,\ldots,z_r).
\]

The adversary matrix will compare a yes-input indexed by \(z\) with a no-input
indexed by \(w\), and it will give larger weight to pairs \((z,w)\) whose paths
agree for longer before they first split.

This is the whole intuition of the construction.  If two paths split very
early, then a successor query near the beginning can distinguish them.  Such
pairs should not be too heavily weighted.  If two paths remain identical for a
long prefix and split only near the end, then distinguishing them requires
learning deeper information about the hidden segment.  These pairs are harder
for a query to separate, and the adversary matrix gives them larger weight.

We now make this precise.

\paragraph{The layered hard family}

Fix integers \(r,M\ge 1\).  Partition the ambient universe into \(r\) layers
\[
  [rM]=V_1\cup V_2\cup\cdots\cup V_r,
  \qquad |V_i|=M.
\]
We identify each layer \(V_i\) with \([M]\).  Thus a tuple
\[
  z=(z_1,\ldots,z_r)\in [M]^r
\]
specifies the hidden segment
\[
  u_{\mathrm{in}}
  \to z_1
  \to z_2
  \to \cdots
  \to z_r
  \to u_{\mathrm{out}},
\]
where \(z_i\) denotes the vertex of layer \(V_i\) with local label \(z_i\). See Figure~\ref{fig:layered-hard-family}.

\begin{figure}[t]
\centering
\begin{tikzpicture}[
  x=1cm,y=1cm,scale=1.6,
  every node/.style={transform shape},
  layerbox/.style={
    draw=Ocean!60,
    fill=SoftBlue,
    rounded corners=6pt,
    line width=.7pt
  },
  candidatenode/.style={
    circle,
    draw=DeadEdge!90,
    fill=NodeBlue!55,
    minimum size=6.6pt,
    inner sep=0pt,
    line width=.7pt
  },
  selectednode/.style={
    circle,
    draw=Ocean!95!black,
    fill=Ocean!35,
    minimum size=8.2pt,
    inner sep=0pt,
    line width=.9pt
  },
  markednode/.style={
    circle,
    draw=MarkRed!95!black,
    fill=SoftRed,
    minimum size=8.6pt,
    inner sep=0pt,
    line width=1pt
  },
  patharrow/.style={
    -{Latex[length=2.5mm]},
    very thick,
    draw=Ocean!90!black
  },
  layerlabel/.style={
    font=\scriptsize,
    text=Ocean!90!black
  },
  vertexlabel/.style={
    font=\scriptsize,
    text=black
  },
  publiclabel/.style={
    font=\small
  }
]

\node[publiclabel] (uin)  at (-0.8,0) {$u_{\rm in}$};
\node[publiclabel] (uout) at (9.55,0) {$u_{\rm out}$};

\def\xA{1.40}
\def\xB{3.00}
\def\xC{4.60}
\def\xD{6.20}
\def\xE{7.80}

\def\ya{0.54}
\def\yb{0.18}
\def\yc{-0.18}
\def\yd{-0.54}

\foreach \x/\lab in {\xA/{1},\xB/{2},\xC/{3},\xD/{r-1},\xE/{r}}{
  \draw[layerbox] (\x-0.42,-0.86) rectangle (\x+0.42,0.86);
  \node[layerlabel] at (\x,1.07) {$V_{\lab}$};
}

\node[text=Ocean!70,font=\large] at (5.40,1.07) {$\cdots$};

\coordinate (v11) at (\xA,\ya);
\coordinate (v12) at (\xA,\yb);
\coordinate (v13) at (\xA,\yc);
\coordinate (v14) at (\xA,\yd);
\node[candidatenode] at (v11) {};
\node[candidatenode] at (v12) {};
\node[candidatenode] at (v13) {};
\node[candidatenode] at (v14) {};

\coordinate (v21) at (\xB,\ya);
\coordinate (v22) at (\xB,\yb);
\coordinate (v23) at (\xB,\yc);
\coordinate (v24) at (\xB,\yd);
\node[candidatenode] at (v21) {};
\node[candidatenode] at (v22) {};
\node[candidatenode] at (v23) {};
\node[candidatenode] at (v24) {};

\coordinate (v31) at (\xC,\ya);
\coordinate (v32) at (\xC,\yb);
\coordinate (v33) at (\xC,\yc);
\coordinate (v34) at (\xC,\yd);
\node[candidatenode] at (v31) {};
\node[candidatenode] at (v32) {};
\node[candidatenode] at (v33) {};
\node[candidatenode] at (v34) {};

\coordinate (v41) at (\xD,\ya);
\coordinate (v42) at (\xD,\yb);
\coordinate (v43) at (\xD,\yc);
\coordinate (v44) at (\xD,\yd);
\node[candidatenode] at (v41) {};
\node[candidatenode] at (v42) {};
\node[candidatenode] at (v43) {};
\node[candidatenode] at (v44) {};

\coordinate (v51) at (\xE,\ya);
\coordinate (v52) at (\xE,\yb);
\coordinate (v53) at (\xE,\yc);
\coordinate (v54) at (\xE,\yd);
\node[candidatenode] at (v51) {};
\node[candidatenode] at (v52) {};
\node[candidatenode] at (v53) {};
\node[candidatenode] at (v54) {};

\node[selectednode, label={[vertexlabel]above left:$z_1$}]     (p1) at (v11) {};
\node[selectednode, label={[vertexlabel]right:$z_2$}]          (p2) at (v22) {};
\node[selectednode, label={[vertexlabel]right:$z_3$}]          (p3) at (v33) {};
\node[selectednode, label={[vertexlabel]right:$z_{r-1}$}]      (p4) at (v44) {};
\node[markednode,   label={[vertexlabel]above right:$z_r$}]    (p5) at (v52) {};

\draw[patharrow] (uin.east) -- (p1);
\draw[patharrow] (p1) -- (p2);
\draw[patharrow] (p2) -- (p3);
\draw[patharrow] (p3) -- (p4);
\draw[patharrow] (p4) -- (p5);
\draw[patharrow] (p5) -- (uout.west);

\node[font=\scriptsize, text=MarkRed!90!black] at (\xE,-1.12) {possible mark};

\end{tikzpicture}

\caption{Each layer $V_i$ contains several possible addresses; the path chooses exactly one of them. All selected vertices except $z_r$ are always unmarked; only $z_r$ may be marked or unmarked.}
\label{fig:layered-hard-family}
\end{figure}
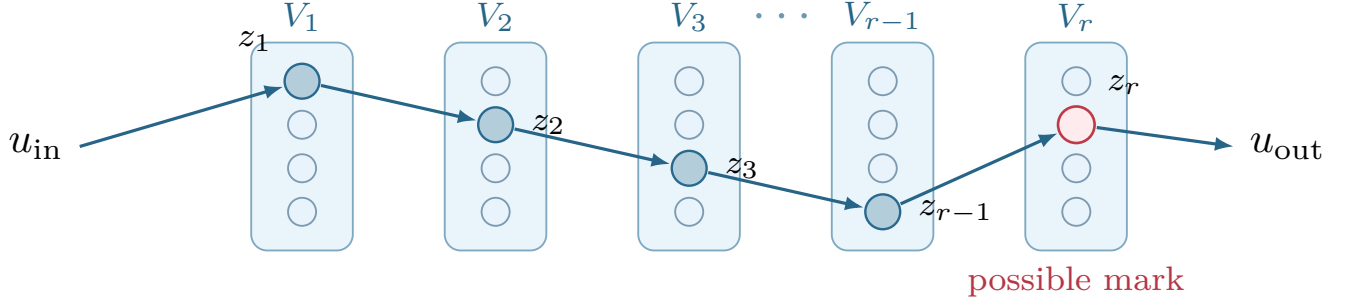

Formally, the successor oracle \(f_z\) is defined by
\[
  f_z(u_{\mathrm{in}})=z_1,
  \qquad
  f_z(z_i)=z_{i+1}\quad (1\le i<r),
  \qquad
  f_z(z_r)=u_{\mathrm{out}},
\]
and
\[
f_z(x)=\bot \quad\text{for every } x\in [rM]\setminus\{z_1,\ldots,z_r\}.
\]

For each tuple \(z=(z_1,\ldots,z_r)\in[M]^r\), define two marking
oracles associated with the same successor oracle \(f_z\).  Let
\(g_z^0\) be the all-zero marking oracle, and let \(g_z^1\) be the marking
oracle defined by
\[
g_z^1(z_r)=1,\qquad g_z^1(x)=0 \text{ for every } x\neq z_r.
\]
We then define two oracle inputs
\[
X_z^0 = (f_z,g_z^0),
\qquad
X_z^1 = (f_z,g_z^1).
\]
Thus \(X_z^0\) is a no-input and \(X_z^1\) is a yes-input.  The restricted
zero- and one-input sets are
\[
\mathcal X_0 = \{X_z^0 : z\in[M]^r\},
\qquad
\mathcal X_1 = \{X_z^1 : z\in[M]^r\}.
\]
For a fixed path label \(z\), the two inputs \(X_z^0\) and \(X_z^1\)
have the same successor oracle \(f_z\) and differ only in whether the
tail vertex \(z_r\) is marked.  

This is a restriction of \(\ASEGDEC(rM,r)\).  The cubic threshold corresponds to
\[
  M=r^2,
  \qquad
  rM=r^3.
\]

\paragraph{The adversary matrix}

Rows of the adversary matrix are indexed by the yes-inputs \(\mathcal X_1\), and columns are indexed by the
no-inputs \(\mathcal X_0\). We identify the row \(X^1_z\) with its path label \(z \in [M]^r\), and the column
\(X^0_w\) with its path label \(w \in [M]^r\). Thus, for the matrices defined below, we write entries
as functions of the path labels \(z\) and \(w\).

For \(m\in[r]\), let \(A^{(m)}\) be the matrix that keeps exactly those pairs
\((z,w)\) whose first disagreement is at coordinate \(m\).  Entrywise,
\[
  A^{(m)}(z,w)=
  \begin{cases}
  M^{-(r-m)}M^{-1/2},
  & z_1=w_1,\ldots,z_{m-1}=w_{m-1}
    \text{ and } z_m\ne w_m,\\
  0,&\text{otherwise.}
  \end{cases}
\]
The full adversary matrix is
\[
  \Gamma_r=\sum_{m=1}^r A^{(m)}.
\]

With the above identification of rows and columns by path labels, \(\Gamma_r(z,w)\) denotes the
entry of \(\Gamma_r\) whose actual row is \(X^1_z\) and whose actual column is \(X^0_w\).
It is useful to rewrite these matrices in a coordinate-by-coordinate form.
Let
\[
  |s\rangle=\frac1{\sqrt M}(1,\ldots,1)^T,
  \qquad
  P=|s\rangle\langle s|,
  \qquad
  B=\frac1{\sqrt M}(J-I),
\]
where the dimension of \(|s\rangle\) is \(M\) and \(J\) is the \(M\times M\) all-ones matrix.  Then
\[
  P(x,y)=\frac1M
  \quad\text{for all }x,y,
\]
while
\[
  B(x,y)=
  \begin{cases}
  1/\sqrt M,&x\ne y,\\
  0,&x=y.
  \end{cases}
\]
Intuitively \(I\) means ``equal at this coordinate,'' \(B\) means ``different at this
coordinate,'' and \(P\) means ``average uniformly over this coordinate.''  With
this notation,
\[
  A^{(m)}
  =
  I^{\otimes(m-1)}
  \otimes B
  \otimes P^{\otimes(r-m)}.
\]
This tensor notation is only a compact way to say the entrywise rule above:
before the first split we require equality, at the first split we require
disagreement, and after the split we average uniformly.

Let us spell out why this tensor expression gives the same entries as the
definition above.  Fix two tuples
\[
  z=(z_1,\ldots,z_r),\qquad w=(w_1,\ldots,w_r).
\]
The \((z,w)\)-entry of
\[
  I^{\otimes(m-1)}\otimes B\otimes P^{\otimes(r-m)}
\]
is the product
\[
  \prod_{i=1}^{m-1} I(z_i,w_i)\cdot B(z_m,w_m)
  \cdot \prod_{i=m+1}^{r} P(z_i,w_i).
\]
The first product is equal to \(1\) exactly when
\[
  z_1=w_1,\ldots,z_{m-1}=w_{m-1},
\]
and is \(0\) otherwise.  The middle factor \(B(z_m,w_m)\) is equal to
\(M^{-1/2}\) exactly when \(z_m\ne w_m\), and is \(0\) otherwise.  Finally,
each suffix factor \(P(z_i,w_i)\) is always equal to \(1/M\), regardless of the
values of \(z_i\) and \(w_i\).  Therefore the product is nonzero exactly when
\(z\) and \(w\) agree before coordinate \(m\) and disagree at coordinate \(m\);
in that case its value is
\[
  M^{-1/2}\cdot M^{-(r-m)}
\]
This is precisely the entrywise definition of \(A^{(m)}(z,w)\).

\paragraph{The numerator of the adversary bound: }

\begin{lemma}
\label{lem:gamma-numerator}
For every \(r,M\ge 1\),
\[
  \|\Gamma_r\|
  \ge
  r\frac{M-1}{\sqrt M}.
\]
\end{lemma}

\begin{proof}
We use the definition of the spectral norm:
\[
  \|\Gamma_r\|=\max_{\|u\|=1}\|\Gamma_ru\|.
\]
Therefore it is enough to find one unit vector that \(\Gamma_r\) stretches by a
large factor.

The natural vector is the uniform vector over all tuples:
\[
  v=|s\rangle^{\otimes r}.
\]
This vector is unit length.  Indeed, it has \(M^r\) entries, each equal to
\(M^{-r/2}\), so
\[
  \|v\|^2
  =
  M^r\cdot M^{-r}
  =
  1.
\]

We now compute the action of one summand \(A^{(m)}\) on \(v\).  The local
matrices satisfy
\[
  I|s\rangle=|s\rangle,
  \qquad
  P|s\rangle=|s\rangle,
\]
and
\[
  B|s\rangle
  =
  \frac1{\sqrt M}(J-I)|s\rangle
  =
  \frac{M-1}{\sqrt M}|s\rangle.
\]
The last identity holds because \(J|s\rangle=M|s\rangle\).

Since
\[
  A^{(m)}
  =
  I^{\otimes(m-1)}
  \otimes B
  \otimes P^{\otimes(r-m)},
\]
all tensor factors except the single \(B\) fix \(|s\rangle\), while the single
\(B\) contributes the scalar \((M-1)/\sqrt M\).  Hence
\[
  A^{(m)}v
  =
  \frac{M-1}{\sqrt M}v
  \qquad
  \text{for every }m\in[r].
\]
Summing over \(m\),
\[
  \Gamma_r v
  =
  \sum_{m=1}^r A^{(m)}v
  =
  r\frac{M-1}{\sqrt M}v.
\]
Since \(\|v\|=1\),
\[
  \|\Gamma_r\|
  \ge
  \|\Gamma_r v\|
  =
  r\frac{M-1}{\sqrt M}.
\]
\end{proof}

\paragraph{The denominator of the adversary bound}
\label{subsubsec:denominator}

We now upper-bound the denominator in the adversary ratio.  Fix a query \(q\).
Recall that \(\Delta_q\) is the \(0\)-\(1\) matrix whose \((z,w)\)-entry is
equal to \(1\) exactly when query \(q\) receives different answers on the
yes-input \(X_z^1\) and the no-input \(X_w^0\).  Thus
\[
  \Gamma_r\circ \Delta_q
\]
is obtained from \(\Gamma_r\) by keeping only those entries that the single
query \(q\) can distinguish.

The important point is that filtering does not create an entirely new kind of
matrix.  Each summand
\[
  A^{(m)}=I^{\otimes(m-1)}\otimes B\otimes P^{\otimes(r-m)}
\]
has a coordinate-by-coordinate structure.  Multiplying entrywise by
\(\Delta_q\) merely refines one or two of these local factors.  For example,
an equality condition \(I\) may become the more restrictive condition that both
coordinates equal a queried value \(a\), while a uniform suffix factor \(P\)
may be replaced by a local matrix that keeps only those row-column pairs on
which the query answers differ.

  We use throughout that the
spectral norm is multiplicative under tensor products:
$  \|R_1\otimes\cdots\otimes R_t\|  =  \prod_{i=1}^t \|R_i\|$. We shall use two norm estimates for the small local matrices that
arise from filtering.  The first estimate, Lemma~\ref{lem:one-coordinate-local-estimates},
covers the one-coordinate matrices \(E_a,F_a,T_a,S_a,K_a\).  These matrices will be defined properly when they first arise inside the proof of
Lemma~\ref{lem:filtered-norms}. Broadly they appear when
a query only restricts one coordinate of the tensor product, for instance by
forcing a coordinate to equal the queried value \(a\), or by keeping only
disagreements involving \(a\).  The second estimate,
Lemma~\ref{lem:two-coordinate-local-estimate}, covers the two-coordinate
matrix \(H_a\).  This matrix appears only for internal successor queries, when
the query depends simultaneously on a vertex and its successor. We just state these lemmas now. Both lemmas
are proved immediately after the proof of Lemma~\ref{lem:filtered-norms}; in Lemma~\ref{lem:filtered-norms} we only use their
stated norm bounds.

\begin{lemma}[One-coordinate local estimates]
\label{lem:one-coordinate-local-estimates}
For \(a\in[M]\), let
\[
  E_a=|a\rangle\langle a|,
  \qquad
  F_a=E_aP,
  \qquad
  T_a=E_aB,
\]
and define
\[
  S_a(x,y)=
  M^{-1/2}\mathbf 1[(x=a)\oplus(y=a)],
\]
\[
  K_a(x,y)=
  \frac1M\mathbf 1[(x=a)\oplus(y=a)].
\]
Then
\[
  \|E_a\|=1,\qquad
  \|P\|=1,\qquad
  \|B\|=\frac{M-1}{\sqrt M},
\]
\[
  \|F_a\|=\frac1{\sqrt M},
  \qquad
  \|T_a\|=\sqrt{\frac{M-1}{M}},
\]
and
\[
  \|S_a\|=\sqrt{\frac{M-1}{M}}\le 1,
  \qquad
  \|K_a\|=\frac{\sqrt{M-1}}{M}\le \frac1{\sqrt M}.
\]
\end{lemma}

\begin{lemma}[The two-coordinate local estimate]
\label{lem:two-coordinate-local-estimate}
Let
\[
  \phi_a(x,y)=
  \begin{cases}
  y,&x=a,\\
  \bot,&x\ne a,
  \end{cases}
\]
and let \(H_a\), indexed by pairs in \([M]^2\), be defined by
\[
  H_a((x,y),(x',y'))
  =
  \frac1{M^2}
  \mathbf 1[\phi_a(x,y)\ne \phi_a(x',y')].
\]
Then
\[
  \|H_a\|\le \frac2{\sqrt M}.
\]
\end{lemma}

\begin{lemma}
\label{lem:filtered-norms}
There is an absolute constant \(C>0\) such that, for every \(r,M\ge 1\),
\[
  \max_q \|\Gamma_r\circ \Delta_q\|
  \le
  C(\sqrt M+r).
\]
\end{lemma}

\begin{proof}
We consider the possible queries.

\paragraph{Mark queries.}

First consider the marking query at address \(a\), that is the query coordinate
\(q=(G,a)\). If \(a \notin V_r\), then \(a\) is never the marked tail in our hard family.  Therefore the answer to \(g(a)\) is \(0\)
on every yes-input and every no-input, and hence
\[
  \Gamma_r\circ \Delta_q=0.
\]

Now suppose \(a\in V_r\), and identify \(a\) with its local label in \([M]\).
In a no-input no vertex is marked.  In a yes-input \(X_z^1\), the unique marked
vertex is the tail \(z_r\).  Hence the query \(g(a)\) distinguishes
\(X_z^1\) from \(X_w^0\) exactly when
\[
  z_r=a.
\]
So the filter keeps precisely those rows whose last coordinate is \(a\).

Let
\[
  E_a=|a\rangle\langle a|.
\]
This is the one-coordinate projector that keeps only the value \(a\) on the row
side.  Define
\[
  F_a=E_aP,
  \qquad
  T_a=E_aB.
\]

If \(m<r\), then the last coordinate lies in the suffix part of \(A^{(m)}\),
where the local factor is \(P\).  Filtering to rows with \(z_r=a\) replaces this
last \(P\)-factor by \(F_a=E_aP\).  Thus
\[
  A^{(m)}\circ \Delta_q
  =
  I^{\otimes(m-1)}
  \otimes B
  \otimes P^{\otimes(r-m-1)}
  \otimes F_a.
\]

By Lemma~\ref{lem:one-coordinate-local-estimates}, 
we get
\[
\|A^{(m)}\circ \Delta_q\| = \|B\| \|F_a\|
= \frac{M - 1}{\sqrt M}\cdot \frac{1}{\sqrt M}
= \frac{M - 1}{M}
\le 1.
\]

If \(m=r\), then the last coordinate is the \(B\)-coordinate.  Filtering to
rows with \(z_r=a\) replaces \(B\) by \(T_a=E_aB\).  Hence
\[
  A^{(r)}\circ \Delta_q
  =
  I^{\otimes(r-1)}\otimes T_a,
\]
and by Lemma~\ref{lem:one-coordinate-local-estimates},
\[
  \|A^{(r)}\circ \Delta_q\|
  =
  \|T_a\|
  \le 1.
\]

Summing over all \(m\),
\[
  \|\Gamma_r\circ \Delta_q\|
  \le
  \sum_{m=1}^r \|A^{(m)}\circ \Delta_q\|
  \le r.
\]

\paragraph{The successor query at the entry symbol.}

Next consider the successor query at the entry symbol,
i.e. the query coordinate \(q=(F,u_{\rm in})\). On input \(z\), this query returns the first vertex
\(z_1\).  Therefore
\[
  \Delta_q(z,w)=1
  \quad\Longleftrightarrow\quad
  z_1\ne w_1.
\]
This is exactly the condition that the first disagreement between \(z\) and
\(w\) occurs at coordinate \(1\).  Thus the filter keeps \(A^{(1)}\) and kills
all \(A^{(m)}\) with \(m>1\):
\[
  \Gamma_r\circ \Delta_q=A^{(1)}.
\]
Therefore
\[
  \|\Gamma_r\circ \Delta_q\|
  =
  \|A^{(1)}\|
  =
  \|B\|
  =
  \frac{M-1}{\sqrt M}
  \le \sqrt M.
\]

\paragraph{Internal successor queries.}

Now fix a successor query to a vertex in an internal layer, namely the
query coordinate
\[
q=(F,v_{j,a}), \qquad 1 \le j < r,
\]
where \(v_{j,a}\) is the vertex in layer \(V_j\) with local label \(a\).

On input \(z\), the query answer depends only on the pair
\[
  (z_j,z_{j+1}).
\]
Indeed, if \(z_j=a\), then the hidden path passes through \(v_{j,a}\), and the
successor is the next-layer vertex \(z_{j+1}\).  If \(z_j\ne a\), then
\(v_{j,a}\) is off the hidden path, and the successor is \(\bot\).  Thus define
\[
  \phi_a(x,y)=
  \begin{cases}
  y,&x=a,\\
  \bot,&x\ne a.
  \end{cases}
\]
Then
\[
  \Delta_q(z,w)=1
  \quad\Longleftrightarrow\quad
  \phi_a(z_j,z_{j+1})\ne \phi_a(w_j,w_{j+1}).
\]

We analyze \(A^{(m)}\circ\Delta_q\) according to the position \(m\) of the
first disagreement.

\medskip
\noindent
\emph{Case 1: \(m>j+1\).}
On the support of \(A^{(m)}\), the tuples agree up to coordinate \(m-1\).
Since \(m>j+1\), this includes both coordinates \(j\) and \(j+1\).  Hence
\[
  z_j=w_j
  \quad\text{and}\quad
  z_{j+1}=w_{j+1}.
\]
The query answers are therefore equal, so the query distinguishes no such
entry:
\[
  A^{(m)}\circ\Delta_q=0.
\]

\medskip
\noindent
\emph{Case 2: \(m=j+1\).}
Here \(z_j=w_j\), but
\[
  z_{j+1}\ne w_{j+1}.
\]
The query distinguishes the two inputs exactly when the common value
\(z_j=w_j\) is \(a\).  Thus the coordinate-\(j\) equality factor \(I\) is
refined to \(E_a\), while the coordinate-\(j+1\) disagreement factor remains
\(B\).  Therefore
\[
  A^{(j+1)}\circ\Delta_q
  =
  I^{\otimes(j-1)}
  \otimes E_a
  \otimes B
  \otimes P^{\otimes(r-j-1)}.
\]
Using \(\|E_a\|=1\) and \(\|B\|=(M-1)/\sqrt M\)   from Lemma~\ref{lem:one-coordinate-local-estimates}, we get
\[
  \|A^{(j+1)}\circ\Delta_q\|
  \le
  \|B\|
  =
  \frac{M-1}{\sqrt M}
  \le \sqrt M.
\]

\medskip
\noindent

\emph{Case 3: \(m=j\).}
Here the first disagreement occurs exactly at the queried coordinate. Thus, on the
support of \(A^{(j)}\), we have
\[
z_1=w_1,\ldots,z_{j-1}=w_{j-1},
\qquad
z_j\neq w_j .
\]
The successor query is to the vertex \(v_{j,a}\). Recall that its answer on a
path \(z\) is
\[
\phi_a(z_j,z_{j+1})
=
\begin{cases}
z_{j+1}, & z_j=a,\\
\bot, & z_j\neq a.
\end{cases}
\]
Therefore, for a pair \((z,w)\) in the support of \(A^{(j)}\), the query
distinguishes the two inputs exactly when
\[
\phi_a(z_j,z_{j+1})\neq \phi_a(w_j,w_{j+1}).
\]
Since \(z_j\neq w_j\), there are three possibilities. If neither \(z_j\) nor
\(w_j\) is equal to \(a\), then both query answers are \(\bot\), so the query
does not distinguish the inputs. If \(z_j=a\) and \(w_j\neq a\), then the two
answers are \(z_{j+1}\) and \(\bot\), respectively, so they differ. Similarly,
if \(z_j\neq a\) and \(w_j=a\), then the two answers are \(\bot\) and
\(w_{j+1}\), respectively, so they differ. Hence, on the support of
\(A^{(j)}\),
\[
\Delta_q(z,w)=1
\quad\Longleftrightarrow\quad
(z_j=a)\oplus(w_j=a).
\]
Before filtering, the \(j\)-th local factor of \(A^{(j)}\) is
$B(x,y)=M^{-1/2}\mathbf 1[x\neq y].$
Multiplication by \(\Delta_q\) keeps precisely those disagreements in which
exactly one side is equal to \(a\). Thus this local factor is replaced by
$S_a(x,y) = M^{-1/2}\mathbf 1[(x=a)\oplus(y=a)].$
All other tensor factors are unchanged. Therefore
\[
A^{(j)}\circ \Delta_q
=
I^{\otimes(j-1)}
\otimes S_a
\otimes P^{\otimes(r-j)}.
\]
By multiplicativity of the spectral norm under tensor products, and using
\(\|I\|=\|P\|=1\), we obtain via Lemma~\ref{lem:one-coordinate-local-estimates},
\[
\|A^{(j)}\circ \Delta_q\|
=
\|I\|^{j-1}\,\|S_a\|\,\|P\|^{r-j}
=
\|S_a\|.
\]
Again by Lemma~\ref{lem:one-coordinate-local-estimates}, \(\|S_a\|\le 1\). Hence
\[
\|A^{(j)}\circ \Delta_q\|\le 1.
\]

\emph{Case 4: \(m<j\).}
Here the first disagreement occurs strictly before the queried layer. Thus the
coordinates \(j\) and \(j+1\) both lie in the suffix part of \(A^{(m)}\). Before
applying the query filter, the contribution of these two coordinates is $P\otimes P.$
Equivalently, if we write the local row pair as
$(x,y)=(z_j,z_{j+1})$
and the local column pair as $(x',y')=(w_j,w_{j+1}),$
then the local \((P\otimes P)\)-entry is
\[
(P\otimes P)\bigl((x,y),(x',y')\bigr)
=
P(x,x')P(y,y')
=
\frac{1}{M^2}.
\]

The query to \(v_{j,a}\) depends only on these two local coordinates. It returns
\[
\phi_a(x,y)
=
\begin{cases}
y, & x=a,\\
\bot, & x\neq a
\end{cases}
\]
on the row input, and
\[
\phi_a(x',y')
=
\begin{cases}
y', & x'=a,\\
\bot, & x'\neq a
\end{cases}
\]
on the column input. Therefore the query distinguishes the two inputs exactly
when
\[
\phi_a(x,y)\neq \phi_a(x',y').
\]
Thus multiplying \(P\otimes P\) by the filter \(\Delta_q\) replaces the local
factor \(P\otimes P\) by the two-coordinate matrix \(H_a\), defined by
\[
H_a\bigl((x,y),(x',y')\bigr)
=
\frac{1}{M^2}
\mathbf 1[\phi_a(x,y)\neq \phi_a(x',y')].
\]
All other tensor factors are unchanged. Hence, for \(m<j\),
\[
A^{(m)}\circ \Delta_q
=
I^{\otimes(m-1)}
\otimes B
\otimes P^{\otimes(j-m-1)}
\otimes H_a
\otimes P^{\otimes(r-j-1)}.
\]
By multiplicativity of the spectral norm under tensor products, and using
\(\|I\|=\|P\|=1\), we get
\[
\|A^{(m)}\circ \Delta_q\|
=
\|B\|\,\|H_a\|.
\]
By Lemma~\ref{lem:one-coordinate-local-estimates},
\[
\|B\|=\frac{M-1}{\sqrt M},
\]
and by Lemma~\ref{lem:two-coordinate-local-estimate},
\[
\|H_a\|\le \frac{2}{\sqrt M}.
\]
Therefore
\[
\|A^{(m)}\circ \Delta_q\|
\le
\frac{M-1}{\sqrt M}\cdot \frac{2}{\sqrt M}
=
\frac{2(M-1)}{M}
\le 2.
\]
Since there are exactly \(j-1\) choices of \(m<j\), the total contribution from
Case~4 is at most
\[
\sum_{m=1}^{j-1}\|A^{(m)}\circ \Delta_q\|
\le
2(j-1).
\]

Combining the four cases for the internal successor query,
\[
  \|\Gamma_r\circ\Delta_q\|
  \le
  \sqrt M+1+2(j-1)
  \le
  \sqrt M+2r.
\]

\paragraph{Successor queries in the last layer.}

Finally consider the successor query in the last layer,
namely the query coordinate
\[
q=(F,v_{r,a}).
\]

On input \(z\), the answer is
\[
  \psi_a(z_r),
  \qquad
  \psi_a(x)=
  \begin{cases}
  u_{\mathrm{out}},&x=a,\\
  \bot,&x\ne a.
  \end{cases}
\]
Thus
\[
  \Delta_q(z,w)=1
\]
exactly when one of \(z_r,w_r\) equals \(a\) and the other does not.

If \(m=r\), then the final coordinate is the \(B\)-coordinate.  The filter keeps
only disagreements involving \(a\), so \(B\) is replaced by \(S_a\):
\[
  A^{(r)}\circ\Delta_q
  =
  I^{\otimes(r-1)}\otimes S_a.
\]
Therefore
\[
  \|A^{(r)}\circ\Delta_q\|
  \le
  \|S_a\|
  \le 1.
\]

If \(m<r\), then the final coordinate lies in the suffix part.  Before
filtering, the last coordinate contributes \(P\).  The filter keeps exactly
those pairs \((x,y)\) for which exactly one of \(x,y\) equals \(a\).  Thus the
last \(P\)-factor is replaced by
\[
  K_a(x,y)=
  \frac1M\mathbf 1[(x=a)\oplus(y=a)].
\]
Hence
\[
  A^{(m)}\circ\Delta_q
  =
  I^{\otimes(m-1)}
  \otimes B
  \otimes P^{\otimes(r-m-1)}
  \otimes K_a.
\]
By Lemma~\ref{lem:one-coordinate-local-estimates},
\[
  \|A^{(m)}\circ\Delta_q\|
  =
  \|B\|\,\|K_a\|
  \le
  \frac{M-1}{\sqrt M}\cdot \frac1{\sqrt M}
  =
  \frac{M-1}{M}
  \le 1.
\]
Summing over \(m\), we get
\[
  \|\Gamma_r\circ\Delta_q\|\le r.
\]

The four query types give the bounds \(r\), \(\sqrt M\), \(\sqrt M+2r\), and
\(r\), respectively.  Therefore
\[
  \max_q\|\Gamma_r\circ\Delta_q\|
  \le
  C(\sqrt M+r)
\]
for an absolute constant \(C>0\).
\end{proof}

We now prove the elementary norm estimates used above, i.e. Lemmas~\ref{lem:one-coordinate-local-estimates} and ~\ref{lem:two-coordinate-local-estimate} respectively.  

\paragraph{Proof of Lemma~\ref{lem:one-coordinate-local-estimates}}
\begin{proof}
The projector \(E_a\) has norm \(1\), and \(P=|s\rangle\langle s|\) is a
rank-one projection, so \(\|P\|=1\).  Also
\[
  B=\frac1{\sqrt M}(J-I).
\]
The matrix \(J-I\) has eigenvalue \(M-1\) on the uniform vector and eigenvalue
\(-1\) on every vector orthogonal to the uniform vector.  Hence
\[
  \|B\|=\frac{M-1}{\sqrt M}.
\]

The matrix \(F_a=E_aP\) has only one nonzero row, namely
\[
  \left(\frac1M,\ldots,\frac1M\right).
\]
The length of this row is \(1/\sqrt M\), so
\[
  \|F_a\|=\frac1{\sqrt M}.
\]
Similarly, \(T_a=E_aB\) has one nonzero row, with \(M-1\) entries equal to
\(1/\sqrt M\).  Therefore
\[
  \|T_a\|=\sqrt{\frac{M-1}{M}}.
\]

After reordering the basis so that \(a\) is first, \(S_a\) has the block form
\[
  S_a=
  \frac1{\sqrt M}
  \begin{pmatrix}
  0 & \mathbf 1^T\\
  \mathbf 1 & 0
  \end{pmatrix},
\]
where \(\mathbf 1\in\mathbb R^{M-1}\) is the all-ones vector.  The unscaled
block matrix has norm \(\|\mathbf 1\|=\sqrt{M-1}\).  Hence
\[
  \|S_a\|=\sqrt{\frac{M-1}{M}}.
\]
The matrix \(K_a\) has the same support pattern, but with scaling \(1/M\)
instead of \(1/\sqrt M\).  Therefore
\[
  \|K_a\|=\frac{\sqrt{M-1}}{M}.
\]
\end{proof}

\paragraph{Proof of Lemma~\ref{lem:two-coordinate-local-estimate}}
\begin{proof}
Partition the local states \((x,y)\in[M]^2\) into inactive states
\[
  \mathcal I=\{(x,y):x\ne a\},
  \qquad |\mathcal I|=M(M-1),
\]
and active states
\[
  \mathcal A=\{(a,y):y\in[M]\},
  \qquad |\mathcal A|=M.
\]
All inactive states have \(\phi_a\)-value \(\bot\).  The active state
\((a,y)\) has \(\phi_a\)-value \(y\).

In the basis ordered as \(\mathcal I,\mathcal A\), the matrix \(H_a\) has the
block form
\[
  H_a
  =
  \frac1{M^2}
  \begin{pmatrix}
  0 & J\\
  J^T & J_M-I_M
  \end{pmatrix},
\]
where \(J\) is the all-ones matrix of size \(M(M-1)\times M\), and \(J_M\) is
the \(M\times M\) all-ones matrix.

By the triangle inequality,
\[
  \|H_a\|
  \le
  \frac1{M^2}
  \left\|
  \begin{pmatrix}
  0 & J\\
  J^T & 0
  \end{pmatrix}
  \right\|
  +
  \frac1{M^2}
  \left\|
  \begin{pmatrix}
  0 & 0\\
  0 & J_M-I_M
  \end{pmatrix}
  \right\|.
\]
The first norm is \(\|J\|\).  Since an all-ones \(p\times q\) matrix has norm
\(\sqrt{pq}\),
\[
  \|J\|
  =
  \sqrt{M(M-1)\cdot M}
  =
  M\sqrt{M-1}.
\]
The second norm is
\[
  \|J_M-I_M\|=M-1.
\]
Therefore
\[
  \|H_a\|
  \le
  \frac{M\sqrt{M-1}}{M^2}
  +
  \frac{M-1}{M^2}
  \le
  \frac1{\sqrt M}+\frac1{\sqrt M}
  =
  \frac2{\sqrt M}.
\]
\end{proof}

\paragraph{Conclusion of the boundary lower bound}

\begin{proof}[Proof of Theorem~\ref{thm:boundary-aseg}]
We first handle \(L=1\).  The problem \(\ASEGDEC(1,1)\) has a unique segment
vertex, and deciding whether it is marked requires \(\Omega(1)\) queries.
Equivalently, by the tightness of the negative-weight adversary bound,
\[
  \Advpm(\ASEGDEC(1,1))=\Omega(1).
\]
Thus the theorem holds for \(L=1\), after choosing the absolute constant small
enough.

Now assume \(L\ge 2\), and write
\[
  r=L,
  \qquad
  M=r^2.
\]
Then
\[
  rM=r^3=L^3.
\]
The layered family above is a restriction of
\[
  \ASEGDEC(r^3,r)=\ASEGDEC(L^3,L).
\]
Therefore, by monotonicity of the adversary bound under restrictions,
\[
  \Advpm(\ASEGDEC(L^3,L))
  \ge
  \frac{\|\Gamma_r\|}
       {\max_q \|\Gamma_r\circ \Delta_q\|}.
\]
By Lemma~\ref{lem:gamma-numerator},
\[
  \|\Gamma_r\|
  \ge
  r\frac{M-1}{\sqrt M}.
\]
Substituting \(M=r^2\), we get
\[
  \|\Gamma_r\|
  \ge
  r\frac{r^2-1}{r}
  =
  r^2-1.
\]
By Lemma~\ref{lem:filtered-norms},
\[
  \max_q \|\Gamma_r\circ \Delta_q\|
  \le
  C(\sqrt M+r)
  =
  C(r+r)
  =
  2Cr.
\]

Hence
\[
  \Advpm(\ASEGDEC(r^3,r))
  \ge
  \frac{r^2-1}{2Cr}.
\]
For \(r\ge 2\), this is \(\Omega(r)\).  Since \(r=L\), there is an absolute
constant \(c_{\mathrm{bd}}>0\) such that
\[
  \Advpm(\ASEGDEC(L^3,L))
  \ge
  c_{\mathrm{bd}}L.
\]
This proves Theorem~\ref{thm:boundary-aseg}.
\end{proof}

\begin{remark}
The proof above also explains why the layered adversary construction is used
only at the cubic threshold and then amplified by composition.

Indeed, suppose one tried to use the same layered construction with \(r\)
layers of width \(M\), without imposing \(M=r^2\).  The numerator estimate from
Lemma~\ref{lem:gamma-numerator} and the denominator estimate from
Lemma~\ref{lem:filtered-norms} would give $ \Advpm
  \ \gtrsim\  \frac{r(M-1)/\sqrt M}{\sqrt M+r}$.
Ignoring constant factors, this behaves like
 $ \frac{r\sqrt M}{\sqrt M+r}$.
Thus the construction gives a linear-in-\(r\) lower bound precisely when $  \sqrt M \gtrsim r,
 \text{ equivalently }M\gtrsim r^2.$
Since the ambient universe of the layered family has size \(rM\), this is
exactly the condition $rM \gtrsim r^3$.

Below this threshold, when \(M<r^2\), the same layered adversary no
longer gives \(\Omega(r)\).  In that regime the denominator is dominated by the
\(r\)-term, and the ratio is only of order  $\sqrt M$.
For example, if one tried to take \(r=\ell\) and \(M\approx N/\ell\), this
direct construction would yield only about $\sqrt{N/\ell}$,
which is much weaker than the desired $(N\ell)^{1/4}$.

This is the reason for the packing argument in
Section~\ref{subsec:one-shot-amplification}.  Instead of forcing one long
\(\ell\)-step segment into a universe too small to support the cubic-threshold
hard instance, we choose smaller segments of length $L_0\approx \sqrt{N/\ell}$,
place each of them at its own cubic threshold \(N_0=L_0^3\), and pack
\(m\approx \ell/L_0\) such segments in series.  The boundary lower bound gives
\(\Omega(L_0)\) for one segment, while the outer unique-OR contributes the
factor \(\Omega(\sqrt m)\).  This recovers $\Omega(\sqrt m\,L_0)
  = \Omega((N\ell)^{1/4})$,
which is the lower bound in the small-universe regime.
\end{remark}

\section{Double Linked List}
\label{app:doubly-linked-list}

This section is dedicated to proving Theorem~\ref{thm:main1}.  The upper
bound is the natural bidirectional version of the algorithm from Section~\ref{sec:uppbound}.  However, the lower bound is
more technical and requires additional ideas. For double linked lists, a predecessor query to a public bridge could
otherwise reveal the last hidden vertex of the preceding segment. This creates issues in the boundary lower bound as the denominator becomes too large. This forces us to use a different type of layered family. As a consequence a more delicate scheme of assigning weights has to be devised for the negative adversary matrix to work. We will discuss these issues later on in Section~\ref{subsec:doubly-lower-bound} when we prove the lower bound.

\subsection{Upper bound} \label{subsec:uppbound}

\begin{theorem}[Upper bound for hidden double linked lists]
\label{thm:dbl-upper}
For all integers $N\geq \ell\geq 1$,
\[
Q(\HDDLDEC(N,\ell)),\ Q(\HDDLSearch(N,\ell))
\in
O\!\left(\min\{\ell,(N\ell)^{1/4}\}\right).
\]
\end{theorem}

\begin{proof}
It suffices to solve the search problem, since the decision problem is obtained by checking
whether the search algorithm outputs a marked vertex.  If $\ell=1$, query $\operatorname{next}(s)$ to obtain
the unique list vertex and then query its mark bit.  Thus, for the rest of the proof, assume
$\ell\geq 2$.

Define
\[
L=\{x\in[N]:\operatorname{next}(x)\neq\bot
\text{ or } \operatorname{prev}(x)\neq\bot\}.
\]
By the promise,
\[
L=\{a_1,\ldots,a_\ell\}.
\]
Indeed, $a_1,\ldots,a_{\ell-1}$ belong to $L$ because their next pointers are non-$\bot$ and
$a_\ell$ belongs to $L$ because $\operatorname{prev}(a_\ell)=a_{\ell-1}$ when $\ell\geq 2$.
Every off-list vertex has both next and prev equal to $\bot$, so no other vertex belongs to $L$.
Hence $|L|=\ell$ is known.  Since membership in $L$ can be checked with a constant number of
oracle queries, Lemma~\ref{lem:exact-good-state-prep} prepares
\[
|L\rangle=\frac{1}{\sqrt{\ell}}\sum_{x\in L}|x\rangle
\]
using $O(\sqrt{N/\ell})$ queries.

\paragraph{Algorithm.}
Fix a parameter $k\in\{1,\ldots,\ell\}$.  For $x\in L$, define the predicate
$\mathrm{Good}_k(x)$ by
\[
\mathrm{Good}_k(x)=1
\quad\Longleftrightarrow\quad
\text{some marked vertex is within list distance at most $k$ from $x$}.
\]
Here distance is measured along the double linked path.

The predicate $\mathrm{Good}_k$ can be evaluated reversibly using $O(k)$ queries.  Starting from
$x$, walk for at most $k$ steps in the $\operatorname{next}$ direction and for at most $k$ steps in the
$\operatorname{prev}$ direction.  Query $g$ on the visited vertices, record whether a marked vertex was seen,
and then reverse the computation to uncompute the workspace.  If either walk reaches $\bot$
before $k$ steps, continue in a fixed dummy state so that the reversible computation has fixed
length.

\begin{center}
\fbox{
\begin{minipage}{0.93\linewidth}
\small
\noindent \textbf{Pseudocode for $\HDDLSearch(N,\ell)$ with parameter $k$}

\medskip

    \noindent \textbf{Input:} Oracle access to $\operatorname{next}, \operatorname{prev}$  and $g$; a parameter $k \in \{1,\ldots,\ell\}$.
\begin{enumerate}
   \item Prepare the state $|L\rangle$ using Lemma~\ref{lem:exact-good-state-prep}.
   \item Run the fixed-point search procedure of Lemma~\ref{lem:fixed-point-search} on the
state $|L\rangle$ with predicate $\mathrm{Good}_k$ and parameter
$\varepsilon=\frac{k}{\ell}$.
   \item If the fixed-point search returns a candidate $x\in L$, verify it by walking from $x$
for at most $k$ steps in both directions and querying $g$ on the visited vertices.
   \item If this verification walk finds a marked vertex, output it.  Otherwise output failure.
\end{enumerate}
\end{minipage}
}
\end{center}

\paragraph{Correctness.}
First consider a no-instance.  Then no list vertex is marked, so no vertex $x\in L$ satisfies
$\mathrm{Good}_k(x)=1$.  Thus the fixed-point search is in its zero-good-weight case, and the
final verification step ensures that the algorithm outputs failure.

Now consider a yes-instance, and let the unique marked vertex be $a_t$.  Every vertex
\[
a_j
\qquad\text{with}\qquad
\max\{1,t-k\}\leq j\leq \min\{\ell,t+k\}.
\]
has list distance at most $k$ from $a_t$, and therefore satisfies $\mathrm{Good}_k(a_j)=1$.

This interval has size at least $k$.  If $t>k$, then it contains
$a_{t-k},a_{t-k+1},\ldots,a_t,$ so it contains at least $k$ vertices.  If $t\leq k$, then it contains
$a_1,\ldots,a_k$ unless the list ends earlier in which case $k = \ell$.  Hence, in a yes-instance,
\[
\frac{|{x\in L:\mathrm{Good}_k(x)=1}|}{|L|}
\geq \frac{k}{\ell}.
\]
Thus Lemma~\ref{lem:fixed-point-search} applies with $\varepsilon=k/\ell$ and outputs, with
constant success probability, a vertex $x\in L$ satisfying $\mathrm{Good}_k(x)=1$.

By the definition of $\mathrm{Good}_k$, once such an $x$ is found, a walk of length at most $k$
from $x$ in one of the two list directions encounters the marked vertex.  Therefore the final
verification and recovery walk outputs the unique marked vertex.  The algorithm therefore solves
$\HDDLSearch(N,\ell)$ with bounded error, and hence also solves $\HDDLDEC(N,\ell)$.

\paragraph{Query complexity.}
Preparing $|L\rangle$ costs $O(\sqrt{N/\ell})$ queries.  One reversible evaluation of
$\mathrm{Good}_k$ costs $O(k)$ queries.  Since the good weight is either zero or at least
$k/\ell$, Lemma~\ref{lem:fixed-point-search} uses
$O\!\left(\sqrt{\frac{\ell}{k}}\right)$
outer iterations.  Each outer iteration uses one preparation of $|L\rangle$ and one evaluation of $\mathrm{Good}_k$.  The final verification and recovery
walk costs an additional $O(k)$ queries.  The
total number of queries is therefore
\[
  O\!\left(k+\left(\sqrt{\frac{N}{\ell}}+k\right)\sqrt{\frac{\ell}{k}}\right).
\]
If $N\leq \ell^3$, choose $k=\lceil\sqrt{N/\ell}\rceil$ and the above expression is
$O((N\ell)^{1/4})$.  If $N>\ell^3$, simply traverse the whole list from $s$ using $\operatorname{next}$ and query
$g$ on every list vertex, using $O(\ell)$ queries.  This proves the stated upper bound for search,
and hence also for decision.
\end{proof}

\subsection{Lower bound}
\label{subsec:doubly-lower-bound}

We now prove the lower bound for hidden double linked lists. 

\begin{theorem}[Lower bound for hidden double linked-list decision]
\label{thm:hddl-lower}
There is an absolute constant \(c>0\) such that, for all integers \(N\ge \ell\ge 1\),
\[
        \mathrm{Adv}^{\pm}(\mathrm{HDDLDEC}(N,\ell))
        \ge
        c\min\{\ell,(N\ell)^{1/4}\}.
\]
Consequently,
\[
        Q(\mathrm{HDDLDEC}(N,\ell)),
        Q(\mathrm{HDDLSearch}(N,\ell))
        \in
        \Omega\!\left(\min\{\ell,(N\ell)^{1/4}\}\right).
\]
\end{theorem}

The proof follows the same two-step strategy as the single linked-list lower bound: first prove a boundary lower
bound for one hard anchored segment at the cubic threshold, and then amplify this boundary
lower bound by serially packing many independent segments and applying adversary composition. The serial composition of independent segments follows the same pattern as in the proof for singly linked lists, however proving the boundary lower bound for double linked lists requires more technicalities and additional ideas.

Let us spell out the obstructions faced for proving the boundary lower bound for double linked lists more
quantitatively.  Suppose that we try to use exactly the anchored 
layered family from the single linked-list lower bound inside a double linked
list.  Thus the hidden path has the form
        $p \to z_1 \to z_2 \to \cdots \to z_{r} \to q$, with $z_r$ either marked or unmarked. The issue in the double linked-list model is the predecessor query at \(q\). This forces the adversary ratio for the matrix to be $O(1)$.

The remedy is to separate the marked vertex from the public exit by a buffer. A buffered
segment consists of a one-sided path $x_1\to x_2\to \cdots \to x_{r},$
followed by an unmarked buffer path $y_1\to y_2\to \cdots \to y_{r}.$ We represent this input by $(x,y)$ with $x = (x_1,x_2, \cdots, x_{r})$ and $y = (y_1, y_2,\cdots,y_{r})$.
The mark, if present, is at \(x_{r}\), but the public
exit sees only the end of the buffer, \(y_{r}\).  This prevents the query \((\operatorname{prev},q)\) from directly revealing the marked vertex.
  
However, the buffer alone is not sufficient if we assign weights in the same way as we did for single linked lists along the entire \(2r\)-layer path.  To see this, suppose we
define the naive first-disagreement matrix with the disagreement at $m$,
        $\widetilde \Gamma
        =
        \sum_{m=1}^{2r} \widetilde A^{(m)},
        \qquad
        \widetilde A^{(m)}
        =
        I^{\otimes(m-1)}\otimes B\otimes P^{\otimes(2r-m)}$,
on the coordinates $x_1,\ldots,x_{r},y_1,\ldots,y_{r}$.
Then, $\|\widetilde\Gamma\|=\Theta(r\sqrt M),$
with an additional constant factor \(2\) hidden in the \(\Theta(\cdot)\).

The problem is now caused by mark queries in the denominator.  Fix an address \(a\) in the layer
\(X_{r}\).  In a yes-input, the only possible marked vertex is
\(x_{r}\), while in a no-input no vertex is marked.  Hence the mark query
\(g(a)\) distinguishes the row indexed by \((x,y)\) from any no-column
indexed by \((u,v)\) exactly when $x_{r}=a$.
Thus the corresponding filter is $ \Delta^g_a((x,y),(u,v)) = \mathbf 1[x_{r}=a]$.

Now consider the summands \(\widetilde A^{(m)}\) with \(m>r\), namely the
summands whose first disagreement occurs somewhere in the buffer. Note that these types of summands did not occur for mark queries for single linked lists.  On the
support of such a summand, all earlier coordinates agree, including the marked
coordinate: $x_{r}=u_{r}$.
After applying the filter \(x_{r}=a\), the local identity factor $I$ at the
marked coordinate is refined to the projector
        $E_a=|a\rangle\langle a|$,
which still has norm \(1\).  

More explicitly, define
       $ v_a
        =
        |s_M\rangle^{\otimes(r-1)}
        \otimes |a\rangle
        \otimes |s_M\rangle^{\otimes r}$.
For every \(m>r\), the filtered summand
\(\widetilde A^{(m)}\circ \Delta^g_a\) acts on \(v_a\) by the scalar
\((M-1)/\sqrt M\). Indeed for \(m>r\), after filtering, the local factors at all coordinates are 
\(I\), \(E_a\), \(B\), or \(P\).  The factors \(I\), \(E_a\), and \(P\) have
norm at most \(1\). However the single \(B\)-coordinate gives,
$B|s_M\rangle
        = \frac{1}{\sqrt M}(J-I)|s_M\rangle = \frac{M-1}{\sqrt {M}}|s_M\rangle$ .

Hence the contribution of the \(r\) buffer-disagreement
summands gives
\[
        \| \widetilde \Gamma\circ \Delta^g_a \|
        \ge
        r\frac{M-1}{\sqrt M}
        =
        \Omega(r\sqrt M).
\]
This is again the same order as the numerator
\(\|\widetilde\Gamma\|=\Theta(r\sqrt M)\).  Thus the adversary ratio of this
naive buffered matrix is still only \(O(1)\).

Thus we see that if the adversary weights are assigned
by first disagreement along the whole buffered path, then many heavily weighted pairs agree at the marked coordinate \(x_{r}\) and first split only later inside the buffer.  A single mark query at \(x_{r}\) sees all such pairs, so the denominator becomes as large as the numerator. We overcome these difficulties by following a different weighting strategy to the matrix $\Gamma$ in Section~\ref{subsubsec:boundbuffdoub}.

In Sections~\ref{subsubsec:buffdoub} to \ref{subsubsec:allbuffdoub} we first show the serial composition of buffered segments works. This is similar to Sections  ~\ref{subsec:anchored-segments} to ~\ref{subsec:serial-packing}. The calculations require some basic modifications, so we give all the proofs for completeness. Then we show the boundary lower bound in Section~\ref{subsubsec:boundbuffdoub} which requires different ideas.

\subsubsection{Buffered anchored segment} \label{subsubsec:buffdoub}

We now define the inner problem used in the lower bound. This plays the same role for double
linked lists that \(\mathrm{ASEGDEC}(N,L)\) played for single linked lists.

\begin{definition}[Buffered anchored double linked segment] \label{def:dbufdec}
Fix integers \(N \ge 2L \ge 1\). Let \(p,q\) be two public connector
symbols, distinct from the private address universe \([N]\).

An input to \(\operatorname{DBUFDEC}(N,L)\) consists of oracle access to
\[
\operatorname{next}:\{p\}\cup [N]\to [N]\cup\{q,\bot\},
\qquad
\operatorname{prev}:[N]\cup\{q\}\to \{p\}\cup [N]\cup\{\bot\},
\qquad
g:[N]\to\{0,1\},
\]
with the promise that there are pairwise distinct private vertices
\[
x_1,\ldots,x_L,y_1,\ldots,y_L\in [N]
\]
such that the hidden double linked path inside the block is
\[
p\to x_1\to x_2\to\cdots\to x_L
\to y_1\to y_2\to\cdots\to y_L\to q .
\]
Thus
\[
\operatorname{next}(p)=x_1,\qquad
\operatorname{next}(x_i)=x_{i+1}\quad (1\le i<L),
\]
\[
\operatorname{next}(x_L)=y_1,\qquad
\operatorname{next}(y_i)=y_{i+1}\quad (1\le i<L),
\qquad
\operatorname{next}(y_L)=q,
\]
and
\[
\operatorname{prev}(x_1)=p,\qquad
\operatorname{prev}(x_i)=x_{i-1}\quad (2\le i\le L),
\]
\[
\operatorname{prev}(y_1)=x_L,\qquad
\operatorname{prev}(y_i)=y_{i-1}\quad (2\le i\le L),
\qquad
\operatorname{prev}(q)=y_L.
\]
Every private address not lying on the displayed path is dead in both directions:
\[
\operatorname{next}(v)=\operatorname{prev}(v)=\bot
\qquad
\text{for every }
v\in [N]\setminus\{x_1,\ldots,x_L,y_1,\ldots,y_L\}.
\]
We will frequently refer to the vertices \(\{x_1,\ldots,x_L\}\) as belonging
to the \(X\)-core part and vertices \(\{y_1,\ldots,y_L\}\) as belonging to
the \(Y\)-buffer part.

The marking promise is that all private vertices are unmarked except possibly
\(x_L\):
\(g(v)=0\) for every  \(v\in [N]\setminus\{x_L\}\).
The output of \(\operatorname{DBUFDEC}(N,L)\) is \(1\) if and only if
\[
g(x_L)=1.
\]
Equivalently, a no-input has no marked vertex, while a yes-input marks exactly
\(x_L\).
\end{definition}

The boundary lower bound for buffered segments is the analogue of Theorem~\ref{thm:boundary-aseg}.

\begin{theorem}[Buffered boundary lower bound]
\label{thm:dbuf-boundary}
There is an absolute constant \(c_{\mathrm{buf}}>0\) such that, for every \(L\ge 1\),
\[
        \mathrm{Adv}^{\pm}\!\left(\mathrm{DBUFDEC}(2L^3,L)\right)
        \ge c_{\mathrm{buf}}L.
\]
\end{theorem}

We defer the proof of Theorem~\ref{thm:dbuf-boundary} to Section~\ref{subsubsec:boundbuffdoub}. We first
show how it implies the full lower bound for hidden double linked lists, i.e. Theorem~\ref{thm:hddl-lower}. 

\subsubsection{Serial packing of buffered segments} \label{subsubsec:serbuffdoub}

We use the same two standard adversary facts as in Section~\ref{sec:lower-bounds}. First, the adversary bound is
monotone under restrictions. Second, if \(F\) is an outer Boolean function and \(G\) is an inner
Boolean function, then the negative adversary bound composes:  $\mathrm{Adv}^{\pm}(F\circ G^m)
        \ge
        \mathrm{Adv}^{\pm}(F)\,\mathrm{Adv}^{\pm}(G).$
We apply this with \(F=\mathrm{UOR}_m\), unique-OR on \(m\) bits, for which $\mathrm{Adv}^{\pm}(\mathrm{UOR}_m)=\Theta(\sqrt m).$

\begin{lemma}
\label{lem:dbuf-serial-packing}
Let \(m,N_0,L_0\ge 1\) with \(N_0\ge 2L_0^3\). There is a restriction of
\[
\operatorname{HDDLDEC}\bigl(mN_0+m+1,\;m(2L_0+1)+1\bigr)
\]
which is exactly
\[
\operatorname{UOR}_m\circ \operatorname{DBUFDEC}(N_0,L_0)^m
\]
under the promise that at most one segment is positive.
\end{lemma}

\begin{proof}
Choose pairwise disjoint private address sets
$U_1,\ldots,U_m$ and $|U_t|=N_0\ (t\in[m])$. Choose fixed connector vertices
$c_0,c_1,\ldots,c_m$
outside \(U_1\cup\cdots\cup U_m\). The public start symbol is \(s\), which is
not an address in \([N]\).

Inside block \(t\), place one copy of \(\operatorname{DBUFDEC}(N_0,L_0)\) on
the private universe \(U_t\), with left connector \(c_{t-1}\) and right
connector \(c_t\). Equivalently, in the notation of
\(\operatorname{DBUFDEC}\), block \(t\) uses
\[
p=c_{t-1},\qquad q=c_t.
\]
If the private segment variables in block \(t\) are
$x_{t,1},\ldots,x_{t,L_0},
y_{t,1},\ldots,y_{t,L_0}$,
then the global hidden double linked list is
$s\to c_0 \to x_{1,1}\to\cdots\to x_{1,L_0}
  \to y_{1,1}\to\cdots\to y_{1,L_0}\to c_1 \to \cdots
\to c_{m-1}
  \to x_{m,1}\to\cdots\to x_{m,L_0}
  \to y_{m,1}\to\cdots\to y_{m,L_0}\to c_m\to\bot $.

Thus the connector \(c_t\) is simultaneously the right connector of block \(t\)
and the left connector of block \(t+1\).

Every address outside the displayed path is dead in both directions.

The list has
$2mL_0+(m+1)=m(2L_0+1)+1$
vertices: \(2L_0\) private vertices from each block and \(m+1\) shared
connector vertices. The ambient address space used has size
$mN_0+(m+1)=mN_0+m+1$.

All connector vertices are forced to be unmarked:
$g(c_t)=0, (0\le t\le m)$.
Inside block \(t\), the marking oracle marks the \(t\)-th
copy of \(\operatorname{DBUFDEC}(N_0,L_0)\).

The only input-dependent parts are the \(m\) private buffered segments. The
shared connector \(c_t\) causes no overlap of input-dependent query entries:
the left block uses the query entry \((\operatorname{prev},c_t)\), while the
right block uses the query entry \((\operatorname{next},c_t)\). These are
different query entries in the double linked-list query model. Therefore the
\(m\) blocks are disjoint as query blocks.

Hence the global list contains a marked vertex if and only if at least one of
the \(m\) buffered segments is positive. Under the promise that at most one
segment is positive, the restricted problem is exactly
\[
\operatorname{UOR}_m\circ \operatorname{DBUFDEC}(N_0,L_0)^m .
\]
\end{proof}

\subsubsection{From the boundary lower bound to all parameters}
\label{subsubsec:allbuffdoub}

This section is the double-linked analogue of the serial-packing step in the
single linked-list lower bound. We will now prove Theorem~\ref{thm:hddl-lower}.

The cases \(\ell\le 5\) are absorbed into the absolute constant, since a fixed
path with one variable mark bit gives a constant lower bound. Assume from now on
that \(\ell\ge 6\).

We split into two regimes.

\paragraph{Regime I: \(N\ge \ell^3\).}
In this regime,
\[
        (N\ell)^{1/4}\ge \ell,
\]
so the desired lower bound is \(\Omega(\ell)\).

Choose
\[
        L_0=\left\lfloor \frac{\ell-2}{2}\right\rfloor,
        \qquad
        N_0=2L_0^3.
\]
Then $2L_0+2\le \ell$ and $L_0=\Omega(\ell)$.
Place one copy of \(\mathrm{DBUFDEC}(N_0,L_0)\) inside the hidden double linked
list. If \(2L_0+2<\ell\), append a fixed unmarked tail after the segment so
that the total list length is exactly \(\ell\). By Lemma~\ref{lem:dbuf-serial-packing}
with \(m=1\), the packed segment uses \(N_0+2\) addresses. Hence the total
number of addresses used after adding the fixed tail is
\[
        N_0+2+\bigl(\ell-(2L_0+2)\bigr)
        =
        2L_0^3+\ell-2L_0.
\]
Since \(L_0\le \ell/2\), this is at most
       $ \frac{\ell^3}{4}+\ell
        \le
        \ell^3$
for all \(\ell\ge 6\). Because \(N\ge \ell^3\), the construction fits inside
the available ambient universe. All unused addresses are padded as dead vertices.

Thus \(\mathrm{HDDLDEC}(N,\ell)\) has a restriction containing one copy of
\(\mathrm{DBUFDEC}(2L_0^3,L_0)\). By monotonicity under restrictions and
Theorem~\ref{thm:dbuf-boundary},
\[
        \mathrm{Adv}^{\pm}(\mathrm{HDDLDEC}(N,\ell))
        \ge
        \mathrm{Adv}^{\pm}(\mathrm{DBUFDEC}(2L_0^3,L_0))
        \ge
        \Omega(L_0)
        =
        \Omega(\ell).
\]
This proves the desired bound in the large-universe regime.

\paragraph{Regime II: \(N < \ell^3\).}
In this regime the target lower bound is
\[
        \Omega\!\left((N\ell)^{1/4}\right).
\]
Put
        $\alpha=\sqrt{\frac{N}{\ell}}$.
Since \(N\ge \ell\), we have \(\alpha\ge 1\), and since \(N < \ell^3\), we have
\(\alpha\le \ell\).

Choose
       $ L_0=\max\left\{1,\left\lfloor \frac{\alpha}{64}\right\rfloor\right\}$,
        $N_0=2L_0^3$,
and
        $m=\max\left\{1,\left\lfloor \frac{\ell-1}{4L_0+2}\right\rfloor\right\}$.
We pack \(m\) buffered segments of parameters \((N_0,L_0)\) using
Lemma~\ref{lem:dbuf-serial-packing}. Thus the packed part has list length $m(2L_0+1)+1$
and uses $mN_0+m+1$
addresses. If needed, we append a fixed unmarked tail to bring the total list
length to exactly \(\ell\).

{\bf Condition 1: We first verify that the packed part has length at most \(\ell\).} If \(m>1\),
then
        $m
        \le
        \frac{\ell-1}{4L_0+2}$,
and hence
\[
        m(2L_0+1)+1
        \le
        \frac{\ell-1}{4L_0+2}(2L_0+1)+1
        =
        \frac{\ell-1}{2}+1
        \le
        \ell.
\]
If \(m=1\) and \(L_0=1\), then
        $m(2L_0+1)+1=4\le \ell$,
since \(\ell\ge 6\). If \(m=1\) and \(L_0>1\), then
        $L_0\le \frac{\alpha}{64}\le \frac{\ell}{64}$,
and therefore
        $m(2L_0+1)+1
        =
        2L_0+2
        \le
        \frac{\ell}{32}+2
        \le
        \ell$.
Thus in all cases,
      $m(2L_0+1)+1\le \ell.$

{\bf Condition 2: the packed segments use at most \(N\) addresses.}
By Lemma~\ref{lem:dbuf-serial-packing}, the packed part uses
        $mN_0+m+1 = m(2L_0^3+1)+1$
addresses.

First suppose \(L_0=1\). Then the packed part uses
       $ 3m+1$
addresses. If \(m>1\), then
        $m\le \frac{\ell-1}{6}$,
so
        $3m+1
        \le
        \frac{\ell-1}{2}+1
        \le
        \ell
        \le
        N.$
If \(m=1\), then \(3m+1=4\le \ell\le N\). Thus the address budget holds when
\(L_0=1\).

Now suppose $L_0>1$. Then
$L_0=\left\lfloor \frac{\alpha}{64}\right\rfloor
\le \frac{\alpha}{64}
\le \frac{\ell}{64}$,
because $\alpha=\sqrt{N/\ell}\le \ell$ in the regime $N\le \ell^3$.
Therefore $\ell/(4L_0)\ge 16$, and in particular
$1\le \ell/(4L_0)$. Moreover,
$\left\lfloor \frac{\ell-1}{4L_0+2}\right\rfloor
\le
\frac{\ell}{4L_0}$.
Since both terms in $m$ are at most
$\ell/(4L_0)$, we get
$m\le \frac{\ell}{4L_0}$.
Therefore
\[
\begin{aligned}
        m(2L_0^3+1)+1
        &\le
        \frac{\ell}{4L_0}(2L_0^3+1)+1  \\
        &=
        \frac{\ell L_0^2}{2}
        +
        \frac{\ell}{4L_0}
        +
        1.
\end{aligned}
\]
Since \(L_0\le \alpha/64\), the first term is at most
       $ \frac{\ell}{2}\left(\frac{\alpha}{64}\right)^2
        =
        \frac{N}{8192}$.
Since \(L_0\ge 2\), the second term is at most
        $\frac{\ell}{8}
        \le
        \frac{N}{8}$.
Finally, because \(L_0 \ge 2\) implies \(\alpha\ge 128\), we have \(N\ge 128^2\ell\ge 6\), so
       $1\le \frac{N}{8}$.
Hence
\[
        m(2L_0^3+1)+1
        \le
        \frac{N}{8192}+\frac{N}{8}+\frac{N}{8}
        \le
        N.
\]

{\bf Condition 3: Now we check that the fixed tail also fits inside the unused address space.}
The remaining list length is
        $\ell-\bigl(m(2L_0+1)+1\bigr)$.
This tail uses one fresh address per added list vertex. Therefore we need
       $N-\bigl(mN_0+m+1\bigr)
        \ge
        \ell-\bigl(m(2L_0+1)+1\bigr)$.
Since \(N_0=2L_0^3\), this is equivalent to
\[
        N-\ell
        \ge
        m(N_0-2L_0)
        =
        m(2L_0^3-2L_0).
\]

If \(L_0=1\), the right-hand side is zero, so the condition follows from
\(N\ge \ell\). Now suppose \(L_0>1\). Then
\[
        L_0=\left\lfloor \frac{\alpha}{64}\right\rfloor\ge 2,
\]
so \(\alpha\ge 128\). Also \(L_0\le \alpha/64\). Since \(L_0\le \ell/64\), we
have
        $\frac{\ell}{4L_0}\ge 16$.
Hence, whether \(m=1\) or \(m>1\), we have
\[
        m\le \frac{\ell}{4L_0}.
\]
Therefore
\[
        m(2L_0^3-2L_0)
        \le
        2mL_0^3
        \le
        \frac{\ell L_0^2}{2}
        \le
        \frac{\ell}{2}\left(\frac{\alpha}{64}\right)^2
        =
        \frac{N}{8192}.
\]
On the other hand,
\[
        N-\ell
        =
        N\left(1-\frac{1}{\alpha^2}\right)
        \ge
        N\left(1-\frac{1}{128^2}\right)
        \ge
        \frac{N}{2}
        \ge
        \frac{N}{8192}.
\]
Thus the fixed tail fits.

So the packed construction gives a valid restriction of
\(\mathrm{HDDLDEC}(N,\ell)\). By Lemma~\ref{lem:dbuf-serial-packing}, this
restriction is
\[
        \mathrm{UOR}_m\circ \mathrm{DBUFDEC}(N_0,L_0)^m.
\]
By monotonicity under restrictions, adversary composition, and
Theorem~\ref{thm:dbuf-boundary},
\[
\begin{aligned}
        \mathrm{Adv}^{\pm}(\mathrm{HDDLDEC}(N,\ell))
        &\ge
        \mathrm{Adv}^{\pm}(\mathrm{UOR}_m)\,
        \mathrm{Adv}^{\pm}(\mathrm{DBUFDEC}(N_0,L_0))  \\
        &=
        \Omega(\sqrt m\,L_0).
\end{aligned}
\]

It remains to compare \(\sqrt m\,L_0\) with \((N\ell)^{1/4}\). Let
        $x=\frac{\ell-1}{4L_0+2}$.
Since
        $m=\max\{1,\lfloor x\rfloor\}$,
we always have
        $m\ge \frac{x}{2}$.
Indeed, if \(x<2\), then \(m\ge 1>x/2\), while if \(x\ge 2\), then
\(\lfloor x\rfloor\ge x/2\). Thus
        $m
        \ge
        \frac{\ell-1}{8L_0+4}$.
Since \(\ell\ge 6\), we have \(\ell-1\ge \ell/2\), and hence
       $m
        \ge
        \frac{\ell}{16(L_0+1)}$.
Therefore
\[
        L_0\sqrt m
        \ge
        L_0\sqrt{\frac{\ell}{16(L_0+1)}}
        =
        \frac{1}{4}
        \sqrt{\frac{\ell L_0^2}{L_0+1}}.
\]
Since \(L_0\ge 1\),
        $\frac{L_0^2}{L_0+1}\ge \frac{L_0}{2}$.
Therefore
\[
        L_0\sqrt m
        \ge
        \frac{1}{4\sqrt 2}\sqrt{\ell L_0}.
\]

Now recall that
\[
        (N\ell)^{1/4}
        =
        \sqrt{\alpha\ell}.
\]
If \(\alpha<128\), then \(L_0=1\), and
\[
        (N\ell)^{1/4}
        =
        \sqrt{\alpha\ell}
        \le
        \sqrt{128}\sqrt{\ell}
        =
        O(L_0\sqrt m).
\]
If \(\alpha\ge 128\), then
\[
        L_0
        =
        \left\lfloor \frac{\alpha}{64}\right\rfloor
        \ge
        \frac{\alpha}{128},
\]
and hence
\[
        L_0\sqrt m
        =
        \Omega(\sqrt{\ell L_0})
        =
        \Omega(\sqrt{\ell\alpha})
        =
        \Omega\!\left((N\ell)^{1/4}\right).
\]
Thus
\[
        \mathrm{Adv}^{\pm}(\mathrm{HDDLDEC}(N,\ell))
        =
        \Omega\!\left((N\ell)^{1/4}\right)
\]
throughout the regime \(N\le \ell^3\).

Combining the two regimes gives
\[
        \mathrm{Adv}^{\pm}(\mathrm{HDDLDEC}(N,\ell))
        \ge
        c\min\{\ell,(N\ell)^{1/4}\}
\]
for an absolute constant \(c>0\).

Finally, the standard adversary characterization gives
\[
        Q(\mathrm{HDDLDEC}(N,\ell))
        =
        \Omega\!\left(\mathrm{Adv}^{\pm}(\mathrm{HDDLDEC}(N,\ell))\right).
\]
Also, any algorithm for \(\mathrm{HDDLSearch}(N,\ell)\) decides whether a marked
vertex exists. Hence the same lower bound holds for
\(\mathrm{HDDLSearch}(N,\ell)\).

\subsubsection{Proof of the buffered boundary lower bound} \label{subsubsec:boundbuffdoub}

It remains to prove Theorem~\ref{thm:dbuf-boundary}.  The proof is a negative-weight
adversary argument.  

Recall the obstruction discussed above.  In the single linked-list boundary
proof, the adversary matrix gives larger weight to a yes/no pair when the two
hidden paths agree for a longer prefix before they first split.  This is the
right intuition there: if two paths split only late, then a query near the
beginning cannot distinguish them, so the pair should be treated as harder.

For double linked lists this idea cannot be applied naively to the whole
buffered path
\[
        p \to x_1 \to x_2 \to \cdots \to x_r
        \to y_1 \to y_2 \to \cdots \to y_r \to q .
\]
The mark, if present, is at the middle vertex \(x_r\), while the public exit
\(q\) sees only the end of the buffer, namely \(y_r\).  The buffer is needed
because otherwise the query \(\mathrm{prev}(q)\) would reveal the marked
endpoint directly.  However, if we simply weight pairs according to the first
disagreement along the entire \(2r\)-vertex path, then many large-weight pairs
would agree throughout the \(X\)-part and first split only inside the buffer.
All such pairs also agree at the marked coordinate \(x_r\).  A single mark
query at \(x_r\) would then see essentially the whole mass of these
buffer-disagreement pairs, making the denominator as large as the numerator.

The construction below avoids this by making the adversary matrix asymmetric
around the middle interface \(x_r \to y_1\).  On the \(X\)-core we use the same matrix from single linked lists: the weight becomes larger when the first disagreement
is closer to \(x_r\).  On the \(Y\)-buffer we use the reversed
matrix: the weight becomes larger when the last disagreement is closer to
\(y_1\), and smaller when disagreements occur near the public exit \(q\). This is analogous to the intuition that if two paths split around the middle then both prev and next queries cannot distinguish them soon, so this pair should be treated as harder.
Thus the largest weights are concentrated around the private middle interface
rather than near the public endpoints.  

We now make this precise. For \(L=1\), the problem has adversary bound \(\Omega(1)\): after fixing the path, a yes-input and a no-input differ in a single mark bit. 
Assume from now on that \(L \ge 2\), and set
       $r = L$ and $M = r^2$.
We prove an \(\Omega(r)\) adversary lower bound for a layered restriction of
\(\operatorname{DBUFDEC}(2r^3,r)\).

\paragraph{The layered buffered hard family.}
Partition the private address universe \([2r^3]\) into \(2r\) layers
\[
        X_1,\ldots,X_r,Y_1,\ldots,Y_r,
        \qquad
        |X_i|=|Y_i|=M .
\]
Since \(M=r^2\), the total number of private addresses is
\[
        2rM = 2r^3 .
\]
We identify every layer with \([M]\).  Hence a pair of strings
\[
        x=(x_1,\ldots,x_r)\in[M]^r,
        \qquad
        y=(y_1,\ldots,y_r)\in[M]^r
\]
specifies the private path
\[
        p \to x_1 \to x_2 \to \cdots \to x_r
        \to y_1 \to y_2 \to \cdots \to y_r \to q,
\]
where \(x_i\) denotes the vertex of layer \(X_i\) with local label \(x_i\),
and \(y_i\) denotes the vertex of layer \(Y_i\) with local label \(y_i\).

For every pair \((x,y)\in[M]^r\times[M]^r\), define two inputs
\[
        X^0_{x,y}
        \qquad\text{and}\qquad
        X^1_{x,y}.
\]
They have the same \(\mathrm{next}\)- and \(\mathrm{prev}\)-oracles, namely the
double linked path above.  The no-input \(X^0_{x,y}\) has no marked vertex.
The yes-input \(X^1_{x,y}\) marks exactly the vertex \(x_r\).  Thus the
restricted zero- and one-input sets are
\[
        \mathcal{X}^0
        =
        \{X^0_{x,y}:x,y\in[M]^r\},
        \qquad
        \mathcal{X}^1
        =
        \{X^1_{x,y}:x,y\in[M]^r\}.
\]
This is a restriction of \(\operatorname{DBUFDEC}(2r^3,r)\).

\paragraph{Construction of the adversary matrix.}

We first fix the disagreement notation.  This is the same notation as in the
single linked-list construction, except that the \(Y\)-buffer is read from the
public exit side back toward the middle.

For strings \(z,w\in[M]^r\), define
\[
        t_X(z,w)
        =
        \min\{i\in[r]:z_i\ne w_i\},
\]
with the convention that \(t_X(z,w)=\infty\) when \(z=w\).  Thus \(t_X(z,w)\)
is the first disagreement in the usual order \(1,2,\ldots,r\).

For buffer strings \(y,v\in[M]^r\), define
\[
        t_Y(y,v)
        =
        \max\{i\in[r]:y_i\ne v_i\},
\]
with the convention that \(t_Y(y,v)=\infty\) when \(y=v\).  Thus
\(t_Y(y,v)=t\) means that \(y\) and \(v\) agree on
\[
        Y_{t+1},Y_{t+2},\ldots,Y_r
\]
and first disagree, when the buffer is read from the public exit side \(q\)
toward the middle, at the layer \(Y_t\).

Let
\[
        |s_M\rangle
        =
        \frac{1}{\sqrt M}\sum_{a=1}^M |a\rangle,
        \qquad
        P=|s_M\rangle\langle s_M|,
        \qquad
        B=M^{-1/2}(J-I),
\]
where the dimension of \(|s\rangle\) is \(M\) and \(J\) is the \(M\times M\) all-ones matrix.

\subparagraph{The matrix on the \(X\)-core.}

For \(1\le t\le r\), define
\[
        A_X^{(t)}
        =
        I^{\otimes(t-1)}
        \otimes B
        \otimes P^{\otimes(r-t)} .
\]
Rows and columns are indexed by strings \(z,w\in[M]^r\).  Entrywise,
\[
        A_X^{(t)}(z,w)
        =
        \begin{cases}
        M^{-(r-t)}M^{-1/2},
        & t_X(z,w)=t,\\
        0,
        & \text{otherwise.}
        \end{cases}
\]
Indeed, \(I^{\otimes(t-1)}\) forces equality before the first split, the
\(B\)-factor enforces disagreement at coordinate \(t\), and the
\(P^{\otimes(r-t)}\)-suffix averages over all coordinates after the split.

Define
\[
        \Gamma^X_r
        =
        \sum_{t=1}^r A_X^{(t)} .
\]
Thus
\[
        \Gamma^X_r(z,w)
        =
        \begin{cases}
        M^{-(r-t_X(z,w))}M^{-1/2},
        & t_X(z,w)<\infty,\\
        0,
        & t_X(z,w)=\infty.
        \end{cases}
\]
This is exactly the adversary matrix from the single linked-list
boundary proof.  Its weight is larger when the first disagreement occurs later,
closer to the marked endpoint \(x_r\).

\subparagraph{The reversed matrix on the \(Y\)-buffer.}

The \(Y\)-buffer uses the same construction, but in the reversed
direction.  Since the original coordinate order is still
\(Y_1,\ldots,Y_r\), the tensor factors are reversed.

For \(1\le t\le r\), define
\[
        A_Y^{(t)}
        =
        P^{\otimes(t-1)}
        \otimes B
        \otimes I^{\otimes(r-t)} .
\]
Rows and columns are indexed by buffer strings \(y,v\in[M]^r\), in the original
coordinate order \(Y_1,\ldots,Y_r\).  Entrywise,
\[
        A_Y^{(t)}(y,v)
        =
        \begin{cases}
        M^{-(t-1)}M^{-1/2},
        & t_Y(y,v)=t,\\
        0,
        & \text{otherwise.}
        \end{cases}
\]
Indeed, the \(I^{\otimes(r-t)}\)-suffix forces agreement on
\(Y_{t+1},\ldots,Y_r\), the \(B\)-factor enforces disagreement at \(Y_t\), and
the \(P^{\otimes(t-1)}\)-prefix averages over \(Y_1,\ldots,Y_{t-1}\), namely
the part lying after the split when the buffer is read from right to left.

Define
\[
        \Lambda^Y_r
        =
        \sum_{t=1}^r A_Y^{(t)} .
\]
Therefore
\[
        \Lambda^Y_r(y,v)
        =
        \begin{cases}
        M^{-(t_Y(y,v)-1)}M^{-1/2},
        & t_Y(y,v)<\infty,\\
        0,
        & t_Y(y,v)=\infty.
        \end{cases}
\]
Thus the buffer weight is largest when the first disagreement from the public
exit side occurs near the middle interface \(y_1\), and smallest when it occurs
near the public exit \(q\).

Finally define the normalized buffer matrix
\[
        H^Y_r
        =
        \frac{\Lambda^Y_r}{\|\Lambda^Y_r\|}.
\]
Then
\[
        \|H^Y_r\|=1.
\]

The normalization keeps the buffer factor at norm one, so the numerator is controlled entirely by the $X$-core matrix and the denominator estimates factor cleanly.

\subparagraph{The final buffered adversary matrix.}

Rows are indexed by yes-inputs \(X^1_{x,y}\), and columns are indexed by
no-inputs \(X^0_{u,v}\).  Equivalently, a row is indexed by
\[
        (x,y)\in[M]^r\times[M]^r,
\]
and a column is indexed by
\[
        (u,v)\in[M]^r\times[M]^r.
\]
The row input marks \(x_r\), while the column input has no marked vertex.

Define
\[
        \Gamma^{\mathrm{buf}}_r
        =
        \Gamma^X_r\otimes H^Y_r .
\]
Thus
\[
        \Gamma^{\mathrm{buf}}_r((x,y),(u,v))
        =
        \Gamma^X_r(x,u)\,H^Y_r(y,v).
\]
Equivalently, if
\[
        s=t_X(x,u),
        \qquad
        t=t_Y(y,v),
\]
then
\[
        \Gamma^{\mathrm{buf}}_r((x,y),(u,v))
        =
        \begin{cases}
        \displaystyle
        \frac{
        M^{-(r-s)}M^{-1/2}\,
        M^{-(t-1)}M^{-1/2}
        }{\|\Lambda^Y_r\|},
        & s,t<\infty,\\[2ex]
        0,
        & s=\infty\text{ or }t=\infty.
        \end{cases}
\]
This is the desired middle-weighted structure.  On the \(X\)-core, the weight
increases as the first disagreement moves toward \(x_r\).  On the \(Y\)-buffer,
the weight increases as the first disagreement from the public exit side moves
toward \(y_1\).  Hence the largest weights are concentrated around the private
middle interface \(x_r\to y_1\), not near either public endpoint. Now let's calculate the norms of the numerator and denominator of this matrix.

\paragraph{The numerator.}

Since
\[
        \Gamma^{\mathrm{buf}}_r
        =
        \Gamma^X_r\otimes H^Y_r,
\]
we have
\[
        \|\Gamma^{\mathrm{buf}}_r\|
        =
        \|\Gamma^X_r\|\,\|H^Y_r\|
        =
        \|\Gamma^X_r\|.
\]
The matrix \(\Gamma^X_r\) is exactly the  matrix from the single
linked-list boundary proof.  Therefore, by Lemma~\ref{lem:gamma-numerator},
\[
        \|\Gamma^X_r\|
        \ge
        r\,\frac{M-1}{\sqrt M}.
\]
Thus
\[
        \|\Gamma^{\mathrm{buf}}_r\|
        \ge
        r\,\frac{M-1}{\sqrt M}.
\]
Substituting \(M=r^2\), we get
\[
        \|\Gamma^{\mathrm{buf}}_r\|
        \ge
        r\,\frac{r^2-1}{r}
        =
        r^2-1
        =
        \Omega(r^2),
\]
because \(r\ge 2\).

\paragraph{The denominator.}
We next prove that every query has filtered norm \(O(r)\).  This is the
denominator estimate for the adversary ratio.

We shall use the following elementary upper bound on the
\(X\)-core matrix:
\[
  \|\Gamma^X_r\|
  \le \sum_{t=1}^r \|A_X^{(t)}\|
  = \sum_{t=1}^r \|B\|
  \le r\sqrt M .
\]
We use,
\[
  A_X^{(t)}=I^{\otimes(t-1)}\otimes B\otimes P^{\otimes(r-t)}
\]
together with \(\|I\|=\|P\|=1\), and the final inequality uses
\(\|B\|=(M-1)/\sqrt M\le \sqrt M\). Check Lemma~\ref{lem:one-coordinate-local-estimates}.

We now state the filtered estimates that will be used in the denominator
argument.  Their proofs are deferred until after the denominator bound of Lemma~\ref{lem:buffered-denominator-bound}. The context in which they are used will also become clearer during the proof of Lemma~\ref{lem:buffered-denominator-bound}.

For \(1\le j\le r-1\) and \(a\in[M]\), define
\[
  \phi_a(u,v)=
  \begin{cases}
    v, & u=a,\\
    \bot, & u\ne a,
  \end{cases}
\]
and for \(2\le j\le r\), define
\[
  \psi_a(u,v)=
  \begin{cases}
    u, & v=a,\\
    \bot, & v\ne a.
  \end{cases}
\]
For strings \(z,w\in[M]^r\), define the following:
\[
\begin{aligned}
  D^X_{\mathrm{in}}(z,w)
    &= \mathbf 1[z_1\ne w_1],\\
  S^{X,+}_{j,a}(z,w)
    &= \mathbf 1[\phi_a(z_j,z_{j+1})\ne \phi_a(w_j,w_{j+1})],
        && 1\le j\leq r-1,\\
  S^{X,-}_{j,a}(z,w)
    &= \mathbf 1[\psi_a(z_{j-1},z_j)\ne \psi_a(w_{j-1},w_j)],
        && 2\le j\le r,\\
  E^X_{1,a}(z,w)
    &= \mathbf 1[(z_1=a)\oplus(w_1=a)],\\
  R^X_a(z,w)
    &= \mathbf 1[z_r=a,\ w_r\ne a],\\
  C^X_a(z,w)
    &= \mathbf 1[z_r\ne a,\ w_r=a],\\
  \Pi^X_a(z,w)
    &= \mathbf 1[z_r=w_r=a],\\
  M^X_a(z,w)
    &= \mathbf 1[z_r=a].
\end{aligned}
\]

\begin{lemma}[Filtered estimates on the \(X\)-core]
\label{lem:X-core-filtered-estimates}
For every \(a\in[M]\), the following estimates hold:
\[
\begin{array}{rclcrcl}
  \|\Gamma^X_r\circ D^X_{\mathrm{in}}\|
    &\le& \sqrt M,
  &&
  \|\Gamma^X_r\circ S^{X,+}_{j,a}\|
    &\le& \sqrt M+2r
       \quad (1\le j<r),
  \\[0.5em]
  \|\Gamma^X_r\circ S^{X,-}_{j,a}\|
    &\le& 2r
       \quad (2\le j\le r),
  &&
  \|\Gamma^X_r\circ E^X_{1,a}\|
    &\le& 1,
  \\[0.5em]
  \|\Gamma^X_r\circ R^X_a\|,
  \ \|\Gamma^X_r\circ C^X_a\|
    &\le& r,
  &&
  \|\Gamma^X_r\circ \Pi^X_a\|
    &\le& \dfrac{r}{\sqrt M},
  \\[0.9em]
  \|\Gamma^X_r\circ M^X_a\|
    &\le& r.
\end{array}
\]
\end{lemma}

\medskip

For buffer strings \(y,v\in[M]^r\), define
\[
\begin{aligned}
  \Delta^Y_R(y,v) &= \mathbf 1[y_r\ne v_r],
  &
  \Delta^Y_L(y,v) &= \mathbf 1[y_1\ne v_1],
  \\[0.3em]
  \Pi^Y_a(y,v) &= \mathbf 1[y_1=v_1=a],
  &
  R^Y_a(y,v) &= \mathbf 1[y_1=a,\ v_1\ne a],
  \\[0.3em]
  C^Y_a(y,v) &= \mathbf 1[y_1\ne a,\ v_1=a],
  &
  T^Y_a(y,v) &= \mathbf 1[(y_r=a)\oplus(v_r=a)].
\end{aligned}
\]

\begin{lemma}[Filtered estimates on the normalized \(Y\)-buffer]
\label{lem:Y-buffer-filtered-estimates}
Assume \(M=r^2\).  For every \(a\in[M]\),
\[
\begin{array}{rclcrcl}
  \|H^Y_r\circ \Delta^Y_R\|
    &=& O(1/r),
  &&
  \|H^Y_r\circ \Pi^Y_a\|
    &=& O(1/M),
  \\[0.5em]
  \|H^Y_r\circ R^Y_a\|,
  \ \|H^Y_r\circ C^Y_a\|
    &=& O(M^{-1/2}),
  &&
  \|H^Y_r\circ \Delta^Y_L\|
    &=& O(1),
  \\[0.5em]
  \|H^Y_r\circ T^Y_a\|
    &=& O\!\left(\frac{1}{r\sqrt M}\right).
\end{array}
\]
Moreover, for any query filter \(\Delta\) coming from an internal
\(Y\)-buffer query, namely from
\(\operatorname{next}(y_{j,a})\) with \(1\le j\le r-1\) or from
\(\operatorname{prev}(y_{j,a})\) with \(2\le j\le r\), we have
\[
  \|H^Y_r\circ \Delta\|
  =
  O\!\left(\frac{\sqrt M+r}{r\sqrt M}\right).
\]
The two boundary queries
\(\operatorname{prev}(y_{1,a})\) and \(\operatorname{next}(y_{r,a})\)
are not included in this estimate and are handled separately.

\end{lemma}

Given these norm estimates let us now prove the bound on the denominator of the adversary matrix.

\begin{lemma}[Denominator bound for the buffered adversary matrix]
\label{lem:buffered-denominator-bound}
There is an absolute constant \(C>0\) such that
\[
  \max_q
  \bigl\|\Gamma^{\mathrm{buf}}_r\circ \Delta_q\bigr\|
  \le Cr,
\]
where \(q\) ranges over all possible queries.
\end{lemma}

\begin{proof}
Recall that
\[
  \Gamma^{\mathrm{buf}}_r=\Gamma^X_r\otimes H^Y_r
\]
and \(\|H^Y_r\|=1\).  We inspect the possible query coordinates.

{\bf Mark queries.}
A mark query can distinguish a yes-input from a no-input only when the queried
address lies in the last \(X\)-layer.  If the queried address has local label
\(a\in X_r\), then the row input is marked exactly when \(x_r=a\).  Hence the
\(X\)-filter is \(M^X_a\), and the \(Y\)-factor is unrestricted.  By
Lemma~\ref{lem:X-core-filtered-estimates},
\[
  \|\Gamma^{\mathrm{buf}}_r\circ \Delta_q\|
  \le
  \|\Gamma^X_r\circ M^X_a\|\,\|H^Y_r\|
  \le r.
\]
All mark queries outside \(X_r\) have zero filtered matrix.

{\bf Connector queries.}
The local connector queries are \(\operatorname{next}(p)\) and
\(\operatorname{prev}(q)\). The outside-neighbor values
\(\operatorname{prev}(p)\) and \(\operatorname{next}(q)\) are not local
query coordinates of \(\mathrm{DBUFDEC}(N,L)\), so they are not considered in
this boundary-segment denominator estimate.

The query \(\operatorname{next}(p)\) reveals \(x_1\). Its \(X\)-filter is
\(D^X_{\mathrm{in}}\), while the \(Y\)-factor is unchanged. Hence, by
Lemma~\ref{lem:X-core-filtered-estimates},
\[
  \|\Gamma^{\mathrm{buf}}_r\circ \Delta_q\|
  \le
  \|\Gamma^X_r\circ D^X_{\mathrm{in}}\|
  \le \sqrt M
  = r,
\]
since \(M=r^2\).

The query \(\operatorname{prev}(q)\) reveals \(y_r\). Therefore
\[
  \Gamma^{\mathrm{buf}}_r\circ \Delta_q
  =
  \Gamma^X_r\otimes (H^Y_r\circ \Delta^Y_R).
\]
Using the elementary bound \(\|\Gamma^X_r\|\le r\sqrt M\) and
Lemma~\ref{lem:Y-buffer-filtered-estimates},
\[
  \|\Gamma^{\mathrm{buf}}_r\circ \Delta_q\|
  \le
  r\sqrt M\cdot O(1/r)
  =
  O(\sqrt M)
  =
  O(r).
\]

{\bf Queries internal to the \(X\)-core.}
First consider a successor query to the vertex with local label \(a\) in
\(X_j\), where \(1\le j \le r-1\).  This query depends on \((x_j,x_{j+1})\), and its
filter on the \(X\)-coordinates is \(S^{X,+}_{j,a}\).  Therefore
\[
  \|\Gamma^{\mathrm{buf}}_r\circ \Delta_q\|
  \le
  \|\Gamma^X_r\circ S^{X,+}_{j,a}\|
  \le
  \sqrt M+2r
  =
  O(r).
\]

Next consider a predecessor query to the vertex with local label \(a\) in
\(X_j\), where \(2\le j\le r\).  This query depends on \((x_{j-1},x_j)\), and
its \(X\)-filter is \(S^{X,-}_{j,a}\).  Hence
\[
  \|\Gamma^{\mathrm{buf}}_r\circ \Delta_q\|
  \le
  \|\Gamma^X_r\circ S^{X,-}_{j,a}\|
  =
  O(r).
\]

Finally, a predecessor query to the vertex with local label \(a\) in \(X_1\)
distinguishes according to the filter \(E^X_{1,a}\).  Thus
\[
  \|\Gamma^{\mathrm{buf}}_r\circ \Delta_q\|
  \le
  \|\Gamma^X_r\circ E^X_{1,a}\|
  \le 1.
\]

{\bf Queries internal to the \(Y\)-buffer, away from the endpoints.}
Let \(\Delta\) be the \(Y\)-filter of such a query.  By
Lemma~\ref{lem:Y-buffer-filtered-estimates},
\[
  \|H^Y_r\circ \Delta\|
  =
  O\!\left(\frac{\sqrt M+r}{r\sqrt M}\right).
\]
Therefore
\[
  \|\Gamma^{\mathrm{buf}}_r\circ \Delta_q\|
  \le
  \|\Gamma^X_r\|\,\|H^Y_r\circ \Delta\|
  \le
  r\sqrt M\cdot
  O\!\left(\frac{\sqrt M+r}{r\sqrt M}\right)
  =
  O(\sqrt M+r)
  =
  O(r).
\]

It remains on the public-exit side to consider the endpoint query
\(\operatorname{next}(a)\) for \(a\in Y_r\).  This query returns \(q\) if
\(y_r=a\), and returns \(\bot\) otherwise.  Hence its \(Y\)-filter is \(T^Y_a\).
By Lemma~\ref{lem:Y-buffer-filtered-estimates},
\[
  \|\Gamma^{\mathrm{buf}}_r\circ \Delta_q\|
  \le
  \|\Gamma^X_r\|\,\|H^Y_r\circ T^Y_a\|
  \le
  r\sqrt M\cdot
  O\!\left(\frac{1}{r\sqrt M}\right)
  =
  O(1).
\]

{\bf The interface query from the $Y$ side.}
Now consider the query \(\operatorname{prev}(a)\) for \(a\in Y_1\).  On the row
input \((x,y)\), the answer is \(x_r\) if \(y_1=a\), and is \(\bot\) otherwise.
Thus the filter is contained in the union of the following
three filters:
\[
  1\otimes R^Y_a,
  \qquad
  1\otimes C^Y_a,
  \qquad
  \Delta^X_{\mathrm{last}}\otimes \Pi^Y_a,
\]
where
\[
  \Delta^X_{\mathrm{last}}(x,u)=\mathbf 1[x_r\ne u_r].
\]

The first two terms have norm at most
\[
  \|\Gamma^X_r\|\,
  O(M^{-1/2})
  \le
  r\sqrt M\cdot O(M^{-1/2})
  =
  O(r).
\]

For the third term, we use a simple monotonicity fact for entrywise nonnegative matrices. Since $\Gamma_r^X$ is entrywise nonnegative and $\Delta_{\mathrm{last}}^X$ is a (0)-(1) filter, we have
$0\le \Gamma_r^X\circ \Delta_{\mathrm{last}}^X \le \Gamma_r^X$
entrywise. Therefore
$||\Gamma_r^X\circ \Delta_{\mathrm{last}}^X||\le ||\Gamma_r^X||$.
Indeed, for any unit vectors $u,w$, writing $|u|$ and $|w|$ for the entrywise absolute-value vectors, we have
$\left|u^T(\Gamma_r^X\circ \Delta_{\mathrm{last}}^X)w\right|
\le
|u|^T(\Gamma_r^X\circ \Delta_{\mathrm{last}}^X)|w|
\le
|u|^T\Gamma_r^X|w|
\le
||\Gamma_r^X||$.
Taking the supremum over (u,w) we have, 
\[||\Gamma_r^X\circ \Delta_{\mathrm{last}}^X||
\le
||\Gamma_r^X||
\le
r\sqrt M .\]

Therefore the third term has norm at most
\[
  \|\Gamma^X_r\|\,\|H^Y_r\circ \Pi^Y_a\|
  \le
  r\sqrt M\cdot O(1/M)
  =
  O\!\left(\frac{r}{\sqrt M}\right)
  =
  O(1),
\]
because \(M=r^2\).  Hence the left interface query has filtered norm \(O(r)\).

{\bf The interface query from the $X$ side.}
Finally consider the query \(\operatorname{next}(a)\) for \(a\in X_r\).  On the
row input \((x,y)\), the answer is \(y_1\) if \(x_r=a\), and is \(\bot\)
otherwise.  Thus the filter is contained in the union of
\[
  R^X_a\otimes 1,
  \qquad
  C^X_a\otimes 1,
  \qquad
  \Pi^X_a\otimes \Delta^Y_L.
\]
By Lemma~\ref{lem:X-core-filtered-estimates},
\[
  \|\Gamma^X_r\circ R^X_a\|,
  \ \|\Gamma^X_r\circ C^X_a\|
  \le r,
\]
so the first two terms contribute \(O(r)\).  For the third term,
Lemma~\ref{lem:X-core-filtered-estimates} and
Lemma~\ref{lem:Y-buffer-filtered-estimates} give
\[
  \|(\Gamma^X_r\circ \Pi^X_a)
      \otimes
      (H^Y_r\circ \Delta^Y_L)\|
  \le
  O\!\left(\frac{r}{\sqrt M}\right)\cdot O(1)
  =
  O(1),
\]
again because \(M=r^2\).  Hence the middle interface query also has filtered
norm \(O(r)\).

Combining all query types, we obtain
\[
  \max_q
  \bigl\|\Gamma^{\mathrm{buf}}_r\circ \Delta_q\bigr\|
  =
  O(r).
\]
This proves the lemma.
\end{proof}

We now prove the two filtered-estimate lemmas.

\begin{proof}[{\bf Proof of Lemma~\ref{lem:X-core-filtered-estimates}}]

Write
\[
  \Gamma^X_r=\sum_{t=1}^r A_X^{(t)},
  \qquad
  A_X^{(t)}
  =
  I^{\otimes(t-1)}\otimes B\otimes P^{\otimes(r-t)}.
\]
We use multiplicativity of the spectral norm under tensor products, the bounds
\(\|I\|=\|P\|=1\), \(\|B\|\le \sqrt M\), and the local estimates from
Lemmas~\ref{lem:one-coordinate-local-estimates} and~\ref{lem:two-coordinate-local-estimate}.

First, consider 
\[
D^X_{\mathrm{in}}(z,w)=\mathbf{1}[z_1\ne w_1].
\]
Every summand \(A_X^{(t)}\) with \(t>1\)
forces \(z_1=w_1\), so it is killed by the filter.  Only \(A_X^{(1)}\)
survives.  Hence
\[
  \|\Gamma^X_r\circ D^X_{\mathrm{in}}\|
  \le
  \|A_X^{(1)}\|
  =
  \|B\|
  \le
  \sqrt M.
\]

Next consider \[
S^{X,+}_{j,a}(z,w)
=
\mathbf{1}\!\left[
\phi_a(z_j,z_{j+1})\ne \phi_a(w_j,w_{j+1})
\right], 1 \le j \le r-1,\] 
the filter for a successor query inside the
\(X\)-core.  If \(t>j+1\), then the coordinates \(j\) and \(j+1\) are both
equality coordinates in \(A_X^{(t)}\), so the query answers are equal and the
contribution is zero.  If \(t=j+1\), the coordinate \(j\) is an equality
coordinate and coordinate \(j+1\) is the \(B\)-coordinate.  The query
distinguishes exactly when the common value at coordinate \(j\) is \(a\), so
the local factor \(I\otimes B\) is refined to \(E_a\otimes B\).  This
contributes at most \(\|E_a\|\|B\|\le \sqrt M\) by Lemma~\ref{lem:one-coordinate-local-estimates}.

If \(t=j\), then the \(B\)-factor is restricted to disagreements in which
exactly one side is equal to \(a\).  Thus \(B\) is replaced by
  $S_a(u,v)=M^{-1/2}\mathbf 1[(u=a)\oplus(v=a)]$,
whose norm is at most \(1\) by Lemma~\ref{lem:one-coordinate-local-estimates}.  This case contributes at most \(1\).
Finally, if \(t<j\), the queried coordinates lie in the \(P\)-suffix.  The
local factor \(P\otimes P\) is replaced by the two-coordinate matrix from
Lemma~\ref{lem:two-coordinate-local-estimate}, whose norm is at most \(2/\sqrt M\).  Multiplication by the single
\(B\)-factor gives contribution at most \(2\) for each such \(t\).  Therefore
\[
  \|\Gamma^X_r\circ S^{X,+}_{j,a}\|
  \le
  \sqrt M+1+2(j-1)
  \le
  \sqrt M+2r.
\]

Now consider \[
S^{X,-}_{j,a}(z,w)=\mathbf{1}\!\left[\psi_a(z_{j-1},z_j)\ne \psi_a(w_{j-1},w_j)\right],\]
the filter for a predecessor query at coordinate
\(j\), where \(2\le j\le r\).  If \(t>j\), the queried coordinates
\(j-1,j\) are both equality coordinates, so the contribution is zero.  If
\(t=j\), the \(B\)-factor is again restricted to disagreements involving
\(a\) by  $S_a(u,v) $, and this contributes at most \(1\) by Lemma~\ref{lem:one-coordinate-local-estimates}.

If \(t=j-1\), then before filtering the local factor on coordinates
\((j-1,j)\) is \(B\otimes P\).  The filter replaces the \(P\)-factor by
\[
  L_a(u,v)=\frac1M\mathbf 1[(u=a)\vee(v=a)].
\]

  Since
\(P\) is the all-\(1/M\) matrix, we have
\[
(E_aP)(u,v)=\frac1M\mathbf 1[u=a],
\qquad
(PE_a)(u,v)=\frac1M\mathbf 1[v=a].
\]
  Since the desired matrix
\(L_a\) should have value \(1/M\), not \(2/M\), at this overlap, we subtract one
copy of the overlap. Thus,
\[
L_a
=
E_aP+PE_a-\frac1M E_a .
\]

and we have \(\|L_a\|\le 3/\sqrt M\) by Lemma~\ref{lem:one-coordinate-local-estimates}.  Multiplying by the \(B\)- factor of $\sqrt{M}$ gives
contribution at most \(3\).  If \(t<j-1\), both queried coordinates lie in the
\(P\)-suffix.  The resulting two-coordinate factor is the coordinate-reversed
analogue of the matrix in Lemma~\ref{lem:two-coordinate-local-estimate}, so it has the same norm bound
\(2/\sqrt M\).  After multiplying by the single \(B\)-factor, each such summand
contributes at most \(2\).  Hence
\[
  \|\Gamma^X_r\circ S^{X,-}_{j,a}\|
  \le
  1+3+2(j-2)
  \le
  2r.
\]

For \[E^X_{1,a} =1[(z_1=a)\oplus(w_1=a)],\] every summand \(A_X^{(t)}\) with \(t>1\) forces
\(z_1=w_1\), so the filter kills it.  For \(t=1\), the \(B\)-factor is replaced
by \(S_a\), whose norm is at most \(1\) by Lemma~\ref{lem:one-coordinate-local-estimates}.  Thus
\[
  \|\Gamma^X_r\circ E^X_{1,a}\|\le 1.
\]

We next prove the estimates for \(R^X_a,C^X_a,\Pi^X_a\). $R^X_a =\mathbf 1[z_r=a,\ w_r\ne a], C^X_a = \mathbf 1[z_r\ne a,\ w_r=a]$ and $\Pi^X_a = \mathbf 1[z_r=w_r=a]$.

The estimates for
\(R^X_a\) and \(C^X_a\) are identical by transposition, so consider
\(R^X_a\).  If \(t=r\), the final \(B\)-factor is restricted to one row, with
the diagonal removed. The entries in row \(a\) are \(1/M\) in columns other than \(a\), and all other entries are zero. This has norm at most \(1\).  If \(t<r\), the final
\(P\)-factor is restricted to row \(a\) and to columns different from \(a\). The row has weights $1/ {M}$ for all positions of the row and 0 elsewhere.
This local factor has norm at most \(1/\sqrt M\).  Multiplication by the single
\(B\)-factor gives contribution at most \(1\).  Summing over \(t\) gives
\[
  \|\Gamma^X_r\circ R^X_a\|,
  \ \|\Gamma^X_r\circ C^X_a\|
  \le r.
\]

For \(\Pi^X_a\), the summand \(t=r\) vanishes because \(B\) has zero diagonal and the filter forces $z_r = w_r = a$.
For each \(t<r\), the final \(P\)-factor becomes \(M^{-1}E_a\), whose norm is
\(1/M\).  Multiplication by the single \(B\)-factor gives contribution at most
\(1/\sqrt M\).  Therefore
\[
  \|\Gamma^X_r\circ \Pi^X_a\|
  \le
  \frac{r}{\sqrt M}.
\]

Finally, consider \(M^X_a(z,w)=\mathbf 1[z_r=a]\).  If \(t=r\), the final
\(B\)-factor is restricted to row \(a\),  giving norm at most \(1\).  If
\(t<r\), the final \(P\)-factor is restricted to row \(a\), which has norm
\(1/\sqrt M\) by Lemma~\ref{lem:one-coordinate-local-estimates}.  Multiplication by the single \(B\)-factor gives contribution at
most \(1\).  Summing over \(t\) gives
\[
  \|\Gamma^X_r\circ M^X_a\|\le r.
\]
This proves all the \(X\)-core estimates.
\end{proof}

\begin{proof}[{\bf Proof of Lemma~\ref{lem:Y-buffer-filtered-estimates}}]

We first prove estimates for the unnormalized matrix
\[
  \Lambda^Y_r=\sum_{t=1}^r A_Y^{(t)},
  \qquad
  A_Y^{(t)}
  =
  P^{\otimes(t-1)}\otimes B\otimes I^{\otimes(r-t)},
\]
and then divide by \(\|\Lambda^Y_r\|\) to get  
$H_r^Y=\frac{\Lambda_r^Y}{\|\Lambda_r^Y\|}.$
We first show that the normalization satisfies $\|\Lambda_r^Y\|=\Omega(r\sqrt M).$
Indeed, let $\rho:[M]^r\to[M]^r$ be the coordinate reversal
$\rho(y_1,\ldots,y_r)=(y_r,\ldots,y_1),$
and let $U_\rho$ be the corresponding permutation matrix. Then
\[
U_\rho A_Y^{(t)}U_\rho^*
=
I^{\otimes(r-t)}\otimes B\otimes P^{\otimes(t-1)}
=
A_X^{(r+1-t)}.
\]
After summing over $t$, this gives
$U_\rho\Lambda_r^YU_\rho^*=\Gamma_r^X$.
Thus $\|\Lambda_r^Y\|=\|\Gamma_r^X\|$. By the numerator estimate for the $X$-core matrix,
\[
\|\Gamma_r^X\|\ge r\frac{M-1}{\sqrt M}=\Omega(r\sqrt M),
\]
since $M=r^2$ and $r\ge 2$.

For  $\Delta^Y_R(y,v) = \mathbf 1[y_r\ne v_r]$  , every summand \(A_Y^{(t)}\) with \(t<r\) is killed, because
the last coordinate \(Y_r\) is then an identity coordinate.  Only \(t=r\) can
survive, and its norm is at most \(\|B\|\le \sqrt M\).  After normalization,
\[
  \|H^Y_r\circ \Delta^Y_R\|=O(1/r).
\]

We now bound the filtered norms one by one. For $\Pi^Y_a(y,v) = \mathbf 1[y_1=v_1=a]$   , the summand \(t=1\) vanishes because the first coordinate is
the \(B\)-coordinate and \(B\) has zero diagonal.  For every \(t>1\), the first
\(P\)-factor becomes \(M^{-1}E_a\), whose norm is \(1/M\), while the later
\(B\)-factor contributes at most \(\sqrt M\).  Thus the unnormalized norm is
\(O(r/\sqrt M)\).  Dividing by \(\Omega(r\sqrt M)\) gives
\[
  \|H^Y_r\circ \Pi^Y_a\|=O(1/M).
\]

For $R^Y_a(y,v)= \mathbf 1[y_1=a,\ v_1\ne a]$   , the summand \(t=1\) restricts the \(B\)-factor to one row, with
the diagonal removed, and hence has norm at most \(1\).  For each \(t>1\), the
first \(P\)-factor is restricted to \(y_1=a\) and \(v_1\ne a\), which has norm
at most \(1/\sqrt M\).  Multiplication by the later \(B\)-factor gives norm at
most \(1\).  Therefore the unnormalized norm is \(O(r)\), and after
normalization,
\[
  \|H^Y_r\circ R^Y_a\|=O(M^{-1/2}).
\]
The estimate for $ C^Y_a(y,v)= \mathbf 1[y_1\ne a,\ v_1=a]   $ is identical by transposition.

For $\Delta^Y_L(y,v) = \mathbf 1[y_1\ne v_1]$ , the summand \(t=1\) contributes at most \(\sqrt M\).  For
each \(t>1\), the first \(P\)-factor is replaced by \(M^{-1}(J-I)\), whose norm
is at most \(1\), and the later \(B\)-factor contributes at most \(\sqrt M\).
Thus the unnormalized norm is \(O(r\sqrt M)\).  After normalization,
\[
  \|H^Y_r\circ \Delta^Y_L\|=O(1).
\]

For $T^Y_a(y,v) = \mathbf 1[(y_r=a)\oplus(v_r=a)]$, every summand \(A_Y^{(t)}\) with \(t<r\) is killed because the
last coordinate \(Y_r\) is then an identity coordinate.  For \(t=r\), the final
\(B\)-factor is restricted to disagreements involving \(a\), whose norm is at
most \(1\).  After normalization,
\[
  \|H^Y_r\circ T^Y_a\|
  =
  O\!\left(\frac{1}{r\sqrt M}\right).
\]

Finally consider a successor or predecessor query wholly inside the $Y$-buffer and not crossing either interface.  We spell out the reversal because this is the only point where the direction of the buffer matters.  Let $\rho(y_1,\ldots,y_r)=(y_r,\ldots,y_1)$
be the coordinate reversal, and let $U_\rho$ be the corresponding permutation matrix.  Since
\[
A_Y^{(t)}=P^{\otimes(t-1)}\otimes B\otimes I^{\otimes(r-t)},
\]
we have
\[
U_\rho A_Y^{(t)}U_\rho^*
=
I^{\otimes(r-t)}\otimes B\otimes P^{\otimes(t-1)}
=
A_X^{(r+1-t)}.
\]
After summing over $t$, this gives $U_\rho\Lambda_r^YU_\rho^*=\Gamma_r^X$.

Now let $\Delta$ be the filter of an internal query in the $Y$-buffer.  First suppose the query is $\mathrm{next}(y_{j,a}), 1\le j<r$.
On a path $y=(y_1,\ldots,y_r)$, this query depends on the pair $(y_j,y_{j+1})$: it returns $y_{j+1}$ if $y_j=a$, and returns $\perp$ otherwise.  After reversal, put $\widetilde y_i=y_{r+1-i}$.
Then $y_j=\widetilde y_{r+1-j}$ and $y_{j+1}=\widetilde y_{r-j}$.  Thus the reversed query asks for the predecessor of the vertex with local label $a$ at coordinate $r+1-j$.  Therefore its reversed filter is precisely one of the $X$-core predecessor filters,
$S_{r+1-j,a}^{X,-}$,
which is covered by Lemma~\ref{lem:X-core-filtered-estimates}.

Similarly, suppose the query is $\mathrm{prev}(y_{j,a}), 2\le j\le r$.
This query depends on $(y_{j-1},y_j)$: it returns $y_{j-1}$ if $y_j=a$, and returns $\perp$ otherwise.  Under the same reversal, $y_j=\widetilde y_{r+1-j}$ and $y_{j-1}=\widetilde y_{r+2-j}$.  Hence the reversed query asks for the successor of the vertex with local label $a$ at coordinate $r+1-j$.  Its reversed filter is therefore one of the $X$-core successor filters,
$S_{r+1-j,a}^{X,+}$,
again covered by Lemma~\ref{lem:X-core-filtered-estimates}.

Equivalently, in both cases there is an internal $X$-core filter $\Delta'$ from Lemma~\ref{lem:X-core-filtered-estimates} such that
$\Delta(y,v)=\Delta'(\rho(y),\rho(v))$
for all $y,v\in[M]^r$.  Hence
$U_\rho(\Lambda_r^Y\circ\Delta)U_\rho^*
=
\Gamma_r^X\circ\Delta'$.
Permutation does not change spectral norm, and Lemma~\ref{lem:X-core-filtered-estimates} gives
\[
\left\|\Lambda_r^Y\circ\Delta\right\|
=
\left\|\Gamma_r^X\circ\Delta'\right\|
=
O(\sqrt M+r).
\]
Finally, since $H_r^Y=\Lambda_r^Y/\left\|\Lambda_r^Y\right\|$ and
$\left\|\Lambda_r^Y\right\|=\Omega(r\sqrt M)$, we get
\[
\left\|H_r^Y\circ\Delta\right\|
=
O\!\left(\frac{\sqrt M+r}{r\sqrt M}\right).
\]
This proves the lemma.
\end{proof}

\paragraph{Conclusion of the buffered boundary lower bound.}
By Lemma~\ref{lem:buffered-denominator-bound},
\[
  \max_q
  \bigl\|\Gamma^{\mathrm{buf}}_r\circ \Delta_q\bigr\|
  =
  O(r).
\]
Together with the numerator estimate
\[
  \|\Gamma^{\mathrm{buf}}_r\|=\Omega(r^2),
\]
this gives
\[
  \frac{\|\Gamma^{\mathrm{buf}}_r\|}
       {\max_q \|\Gamma^{\mathrm{buf}}_r\circ \Delta_q\|}
  =
  \Omega(r).
\]
The layered buffered family is a restriction of
\[
  \mathrm{DBUFDEC}(2r^3,r).
\]
Therefore, by monotonicity of the adversary bound under restrictions,
\[
  \mathrm{Adv}^{\pm}\!\left(\mathrm{DBUFDEC}(2r^3,r)\right)
  =
  \Omega(r).
\]
Since \(r=L\), and since the case \(L=1\) was already handled separately, there
is an absolute constant \(c_{\mathrm{buf}}>0\) such that, for every \(L\ge1\),
\[
  \mathrm{Adv}^{\pm}\!\left(\mathrm{DBUFDEC}(2L^3,L)\right)
  \ge
  c_{\mathrm{buf}}L.
\]
This proves Theorem~\ref{thm:dbuf-boundary}.

\section*{Acknowledgements} 

ChatGPT 5.5 has been used for calculating the matrix norms of Lemmas~\ref{lem:one-coordinate-local-estimates},~\ref{lem:two-coordinate-local-estimate},~\ref{lem:X-core-filtered-estimates} and~\ref{lem:Y-buffer-filtered-estimates}. All calculation results have been independently verified by the authors. It has also been used for reviewing the draft and fixing minor technical errors. 

This work was supported by
JSPS Grant-in-Aid for Scientific Research (S) Nos.~24H00071,~25K24674, (A) Nos.~23H00468, 26H02489, (C) Nos.~25K15115, 25K14992, JSPS Grant-in-Aid for Challenging Research (Pioneering) No.~23K17455, and MEXT Quantum Leap Flagship Program (MEXT Q-LEAP) Grant Number JPMXS0120319794, JST CRONOS Japan Grant Number JPMJCS24K2.

\bibliography{refs}

\section{Appendix}

\subsection{Proof of Proposition~\ref{prop:classical-baseline}} \label{app:propclass}
\begin{proof}
The upper bound is immediate. Starting from the public symbol $s$, one can follow either the
successor pointers (predecessor pointers not needed) along the hidden path and query the marking oracle on each visited vertex.
This uses $O(\ell)$ queries and solves the problem.

For the lower bound, it suffices to consider a restriction that makes the ambient embedding
completely public and therefore useless. Fix the hidden path to be the public chain
$f(s)=1,
f(i)=i+1\ (1\le i<\ell),
f(\ell)=\pp,
f(x)=\pp\ (x\in [N]\setminus[\ell])$,
and let the only variable part of the input be the marking string
$(g(1),\dots,g(\ell))\in\{0,1\}^{\ell}$
under the promise $\sum_{i=1}^{\ell} g(i)\le 1$.
On this restricted domain, both $\HLLDEC(N,\ell)$ and $\HDDLDEC(N,\ell)$ are exactly the unique-OR problem on $\ell$
bits. It is a folklore result that both deterministic and randomized query complexity of unique-OR is $\Theta(\ell)$. We nevertheless include a proof for completeness. 

For deterministic algorithms, consider the all-zero marking string $0^\ell$. If an algorithm
queries fewer than $\ell$ mark bits on input $0^\ell$, then some index $i\in[\ell]$ is never
queried. The executions on $0^\ell$ and on the input $e_i$ with a unique $1$ at position $i$
are then identical, so the algorithm cannot be correct on both. Hence, both $D(\HLLDEC(N,\ell)), D(\HDDLDEC(N,\ell)) \ge \ell$.

For bounded-error randomized algorithms, apply Yao's minimax principle~\cite{Yao1977}.
Let $\mu$ be the distribution that chooses $0^\ell$ with probability $1/2$ and, for each
$i\in[\ell]$, chooses $e_i$ with probability $1/(2\ell)$. Consider a deterministic algorithm
$A$ making at most $T$ queries. Let $S\subseteq[\ell]$ be the set of indices queried by $A$
on input $0^\ell$; then $|S|\le T$. If $A$ outputs $1$ on $0^\ell$, then its error under $\mu$
is already at least $1/2$. Otherwise $A$ outputs $0$ on $0^\ell$, and for every $i\notin S$
the entire execution of $A$ on $e_i$ is identical to its execution on $0^\ell$, since all queried
positions are $0$. Thus $A$ also outputs $0$ on $e_i$ and errs. Therefore
\[
\Pr_{x\sim\mu}[A\text{ is correct on }x]
\le
\frac12+\frac{|S|}{2\ell}
\le
\frac12+\frac{T}{2\ell}.
\]
To achieve success probability at least $2/3$, one must have $T\ge \ell/3$. Hence, both $R(\HLLDEC(N,\ell))$ and $ R(\HDDLDEC(N,\ell))\in \Omega(\ell)$.

The lower bound on the decision version immediately also yields, 
\[R(\HLLSearch(N,\ell)), R(\HDDLSearch(N,\ell))
\in \Omega(\ell).\]
This proves the claim.
\end{proof}

\subsection{Model and Proof of Hidden Ordered Search} \label{app:ordsearch}

Concretely, fix integers \(N \ge n\). There is an ambient universe \([N]\), a
hidden support set \(S \subseteq [N]\) of size \(n\), and a hidden balanced
binary search tree on \(S\), specified by child-pointer oracles
$L,R : \{s\}\cup [N] \to [N]\cup\{\perp\}.$
There is a distinguished root \(r\in S\) such that w.l.o.g.,
$L(s)=r, R(s)=\perp.$
The non-\(\perp\) structure reachable from \(r\) by the child pointers
\(L\) and \(R\) is exactly the hidden balanced binary search tree on \(S\).
For every \(x\in [N]\setminus S\), we have $L(x)=R(x)=\perp.$
Let \(v_1,\ldots,v_n\) be the support vertices in their in-order order.

In addition, there is a comparison
oracle $c : [N]\to\{0,1\}$
such that \(c(x)=0\) off the support and, on the support,
$0=x_0 \le c(v_1) \le c(v_2) \le \cdots \le c(v_n)=1.$
The task is to output the vertex
\[
v_{i^\star}, i^\star = \min\{i\in[n]: c(v_i)=1\}.
\]
We call this problem hidden ordered search.

\begin{proposition} \label{prop:orderedsearch}
The hidden-support ordered-search problem described above has quantum query complexity
$\Theta(\log n)$.
\end{proposition}

\begin{proof} \label{app:proporder}

For the upper bound, query \(L(s)\) to obtain the root \(r\), and then follow the usual comparison path in the balanced binary search tree. At each visited support vertex, query the comparison bit and then move to the appropriate
child. Since the tree is balanced, this path has length \(O(\log n)\). Hence the problem is solvable using $O(\log n)$ queries.

For the lower bound, restrict to the subfamily in which the support set, the
balanced binary search tree, and all child pointers are fixed and public. Let $v_1,\ldots,v_n$
be the support vertices in sorted order. The only variable part of the input is
the comparison string $x_i= c(v_i), i\in[n],$
under the promise $0=x_0 \le x_1 \le \cdots \le x_n = 1.$
On this restricted subfamily, the task is exactly ordinary ordered search: output $i^\star = \min\{i\in[n]:x_i=1\}.$
Ordinary ordered search has bounded-error quantum query complexity
\(\Omega(\log n)\) \cite{HNS02,BDW99}. Since query complexity can only decrease
under restriction, the original hidden-support ordered-search problem also has
quantum query complexity \(\Omega(\log n)\). Combining the upper and lower bounds gives \(\Theta(\log n)\).

\end{proof}

Therefore, merely placing a structured search problem inside a larger ambient universe does
not by itself produce a new asymptotic quantum speedup. In ordered search, the comparison
structure still funnels the algorithm along a logarithmic decision tree, so the ambient
parameter is asymptotically irrelevant.

\end{document}